\documentclass[11pt,draftcls,onecolumn]{IEEEtran}
\usepackage{graphicx}
\usepackage{amsmath}
\usepackage{mathrsfs}
\usepackage{mathtools}
\usepackage{amssymb}
\usepackage{float}
\usepackage{url}
\usepackage{verbatim}
\usepackage{dsfont}
\usepackage{enumerate}
\usepackage{layout}

\usepackage[left=20mm,right=20mm,top=24mm,bottom=24mm]{geometry}

\newtheorem{theorem}{Theorem}
\newtheorem{lemma}{Lemma}
\newtheorem{proposition}{Proposition}

\newtheorem{definition}{Definition}
\newtheorem{example}{Example}
\newtheorem{remark}{Remark}

\newcommand{\naturals}{\ensuremath{\mathbb{N}}}
\newcommand{\Reals}{\ensuremath{\mathbb{R}}}
\newcommand{\expectation}{\ensuremath{\mathbb{E}}}
\newcommand{\set}{\ensuremath{\mathcal}}

\begin{document}
\title{\huge{Tight Bounds on the R\'{e}nyi Entropy via Majorization\\
with Applications to Guessing and Compression}\footnote{The paper
was published in the {\em Entropy} journal (special issue on {\em Probabilistic
Methods in Information Theory, Hypothesis Testing, and Coding}), vol.~20, no.~12,
paper no.~896, November~22, 2018. Online available at \url{https://www.mdpi.com/1099-4300/20/12/896}.}}
\author{\vspace*{0.2cm} Igal Sason\\[-0.1cm]
Department of Electrical Engineering\\[-0.1cm]
Technion--Israel Institute of Technology\\[-0.1cm]
Haifa 3200003, Israel\\[-0.1cm]
E-mail: \url{sason@ee.technion.ac.il}.}

\maketitle

\begin{abstract}
This paper provides tight bounds on the R\'{e}nyi entropy of a function of a
discrete random variable with a finite number of possible values, where the
considered function is not one-to-one. To that end, a tight lower bound on
the R\'{e}nyi entropy of a discrete random variable with a finite support is
derived as a function of the size of the support, and the ratio of the
maximal to minimal probability masses. This work was inspired by the recently
published paper by Cicalese {\em et al.}, which is focused on the Shannon entropy,
and it strengthens and generalizes the results of that paper to R\'{e}nyi
entropies of arbitrary positive orders.
In view of these generalized bounds and the works by
Arikan and Campbell, non-asymptotic bounds are derived for guessing moments
and lossless data compression of discrete memoryless sources.
\end{abstract}


\eject

\section{Introduction}
\label{section: Introduction}

Majorization theory is a simple and productive concept in the
theory of inequalities, which also unifies a variety of familiar bounds \cite{HLP52,MarshallOA}.
The concept of majorization finds various applications in diverse fields (see, e.g., \cite{Arnold07})
such as economics \cite{ArnoldS18,CicaleseGV13,MarshallOA}, combinatorial analysis
\cite{MarshallOA,Steele}, geometric inequalities \cite{MarshallOA}, matrix theory
\cite{Bhatia,HornJ, MarshallOA,Steele}, Shannon theory \cite{Ben-Bassat-Raviv,CicaleseGV04,
CicaleseGV13,CicaleseGV17,CicaleseGV18,CicaleseV_ISIT18,Harremoes04,HoY_IT2010,HoSV_IT2010,
HoSV-ISIT15,Joe88,Joe90,Koga13,Puchala13,SasonV18a,verdubook,Witsenhhausen,XiWZ11}, and wireless
communications \cite{Hanly2012,FnT06a,FnT06b,Roventa2015,SezginJ10,Pramod1999,Pramod1999b,Pramod2002}.

This work, which relies on the majorization theory, has been greatly inspired by the
recent insightful paper by Cicalese {\em et al.} \cite{CicaleseGV18}.\footnote{The
research work in the present paper has been initialized while the author handled \cite{CicaleseGV18}
as an associate editor.} The work in \cite{CicaleseGV18}
provides tight bounds on the Shannon entropy of a function of a discrete random variable with a finite
number of possible values, where the considered function is not one-to-one. For that purpose,
and while being of interest by its own right (see \cite[Section~6]{CicaleseGV18}), a tight
lower bound on the Shannon entropy of a discrete random variable with a finite support was
derived in \cite{CicaleseGV18} as a function of the size of the support, and the ratio of the
maximal to minimal probability masses. The present paper aims to extend the bounds in \cite{CicaleseGV18}
to R\'{e}nyi entropies of arbitrary positive orders (note that the Shannon entropy is equal to
the R\'{e}nyi entropy of order~1), and to study the
information-theoretic applications of these (non-trivial) generalizations in the context of
non-asymptotic analysis of guessing moments and lossless data compression.

The motivation for this work is rooted in the diverse information-theoretic applications of R\'{e}nyi
measures \cite{Renyientropy}. These include (but are not limited to) asymptotically tight bounds on guessing moments
\cite{Arikan96}, information-theoretic applications such as guessing subject to distortion \cite{ArikanM98-1},
joint source-channel coding and guessing with application to sequential decoding \cite{ArikanM98-2},
guessing with a prior access to a malicious oracle \cite{BurinS_IT18}, guessing while allowing
the guesser to give up and declare an error \cite{Kuzuoka18}, guessing in secrecy problems
\cite{MerhavAr99,Sundaresan07b}, guessing with limited memory \cite{Salam17}, and guessing
under source uncertainty \cite{Sundaresan07}; encoding tasks \cite{BracherLP17,BunteL14a};
Bayesian hypothesis testing \cite{Ben-Bassat-Raviv,SasonV18a,verdubook},
and composite hypothesis testing \cite{Shayevitz_ISIT11,TomamichelH16};
R\'{e}nyi generalizations of the rejection sampling problem in \cite{Harsha10}, motivated
by the communication complexity in distributed channel simulation, where these generalizations distinguish
between causal and non-causal sampler scenarios \cite{LiuSV18}; Wyner's common information
in distributed source simulation under R\'{e}nyi divergence measures \cite{YuT18};
various other source coding theorems \cite{Campbell65,CV2014a,CV2014b,HT18,Kuzuoka16,
Kuzuoka18,SasonV18b,TanH18,Tyagi17,verdubook}, channel coding theorems \cite{Arimoto73,
Arimoto75,Csiszar95,PolyanskiyV10,Sason16,Tyagi17,verdubook,YuT17}, including coding theorems
in quantum information theory \cite{Dalai13,Leditzky16,MosonyiO15}.

The presentation in this paper is structured as follows:
Section~\ref{section: notation} provides notation and essential preliminaries
for the analysis in this paper.
Sections~\ref{section:thm-RE} and~\ref{section: Function of RV} strengthen
and generalize, in a non-trivial way, the bounds on the Shannon entropy in \cite{CicaleseGV18}
to R\'{e}nyi entropies of arbitrary positive orders (see Theorems~\ref{theorem: c}
and~\ref{thm:Huffman}).
Section~\ref{section: IT Applications} relies on the generalized bound from
Section~\ref{section: Function of RV} and the work by Arikan \cite{Arikan96} to derive
non-asymptotic bounds for guessing moments (see Theorem~\ref{thm: guessing});
Section~\ref{section: IT Applications} also relies on the generalized bound in Section~\ref{section: Function of RV}
and the source coding theorem by Campbell \cite{Campbell65} (see Theorem~\ref{theorem: Campbell})
for the derivation of non-asymptotic bounds for lossless compression of discrete memoryless sources
(see Theorem~\ref{thm: lossless compression}).

\section{Notation and Preliminaries}
\label{section: notation}
Let
\begin{itemize}
\item $P$ be a probability mass function defined on a finite set $\set{X}$;
\item $p_{\max}$ and $p_{\min}$ be, respectively, the maximal and minimal positive masses of $P$;
\item $G_P(k)$ be the sum of the $k$ largest masses of $P$ for $k \in \{1, \ldots, |\set{X}|\}$
(note that $G_P(1) = p_{\max}$ and $G_P(|\set{X}|) =1$);
\item $\set{P}_n$, for an integer $n \geq 2$, be the set of all probability mass functions defined on
$\set{X}$ with $|\set{X}|=n$; without any loss of generality, let $\set{X} = \{1, \ldots, n\}$;
\item $\set{P}_n(\rho)$, for $\rho \geq 1$ and an integer $n \geq 2$, be the subset of all probability
measures $P \in \set{P}_n$ such that
\begin{align}
\label{eq: ratio max/min}
\frac{p_{\max}}{p_{\min}} \leq \rho.
\end{align}
\end{itemize}

\begin{definition}[\em{majorization}]
\label{definition: majorization}
Consider discrete probability mass functions $P$ and $Q$ defined on the same
(finite or countably infinite) set $\set{X}$. It is said that $P$ is {\em majorized} by $Q$
(or $Q$ majorizes $P$), and it is denoted by $P \prec Q$, if $G_P(k) \leq G_Q(k)$ for
all $k \in \{1, \ldots, |\set{X}|-1\}$ (recall that $G_P(|\set{X}|)=G_Q(|\set{X}|)=1$).
If $P$ and $Q$ are defined on finite sets of different cardinalities, then the
probability mass function which is defined over the smaller set is first padded
by zeros for making the cardinalities of these sets be equal.
\end{definition}

By Definition~\ref{definition: majorization}, a unit mass majorizes
any other distribution; on the other hand, the equiprobable distribution on a finite
set is majorized by any other distribution defined on the same set.

\begin{definition}[\em{Schur-convexity/concavity}]
A function $f \colon \set{P}_n \to \Reals$ is said to be {\em Schur-convex} if for
every $P,Q \in \set{P}_n$ such that $P \prec Q$, we have $f(P) \leq f(Q)$. Likewise,
$f$ is said to be {\em Schur-concave} if $-f$ is Schur-convex, i.e., $P,Q \in \set{P}_n$
and $P \prec Q$ imply that $f(P) \geq f(Q)$.
\end{definition}

\begin{definition}[\em{R\'{e}nyi entropy} \cite{Renyientropy}]
\label{definition: Renyi entropy}
Let $X$ be a random variable taking values on a finite
or countably infinite set $\set{X}$, and let $P_X$ be its probability
mass function. The R\'{e}nyi entropy of order
$\alpha \in (0,1) \cup (1, \infty)$
is given by\footnote{Unless explicitly stated, the logarithm base
can be chosen by the reader, with $\exp$ indicating the inverse function
of $\log$.}
\begin{align} \label{eq: Renyi entropy}
H_{\alpha}(X) & = H_{\alpha}(P_X) = \frac1{1-\alpha} \, \log
\left( \, \sum_{x \in \set{X}} P_X^{\alpha}(x)\right).
\end{align}
By its continuous extension,
\begin{align}
\label{eq: RE of zero order}
& H_0(X) = \log \, \bigl| \{x \in \set{X} \colon
P_X(x) > 0 \} \bigr|, \\
\label{eq: Shannon entropy}
& H_1(X) = H(X), \\
\label{eq: RE of infinite order}
& H_{\infty}(X) = \log \frac1{p_{\max}}
\end{align}
where $H(X)$ is the (Shannon) entropy of $X$.
\end{definition}

\begin{proposition}[Schur-concavity of the R\'{e}nyi entropy (Appendix F.3.a (p.~562) of \cite{MarshallOA})]
\label{proposition: Renyi entropy is Schur-concave}
The R\'{e}nyi entropy of an arbitrary order $\alpha > 0$ is Schur-concave;
in particular, for $\alpha=1$, the Shannon entropy is Schur-concave.
\end{proposition}

\begin{remark}
\label{remark: beyond Prop. 1}
\cite[Theorem~2]{HoSV-ISIT15} strengthens Proposition~\ref{proposition: Renyi entropy is Schur-concave},
though it is not needed for our analysis.
\end{remark}

\begin{definition}[\em{R\'{e}nyi divergence} \cite{Renyientropy}]
Let $P$ and $Q$ be probability mass functions defined on a finite or countably infinite set $\set{X}$.
The {\em R\'{e}nyi divergence of order $\alpha \in [0, \infty]$} is defined as follows:
\begin{itemize}
\item
If $\alpha \in (0,1) \cup (1, \infty) $, then
\begin{align} \label{eq:RD0}
D_{\alpha}(P\|Q) = \frac1{\alpha-1} \; \log \,
\sum_{x \in \set{X}} P^\alpha (x) \,
Q^{1-\alpha} (x).
\end{align}
\item By the continuous extension of $D_{\alpha}(P \| Q)$,
\begin{align}
\label{eq: d0}
D_0(P \| Q) &= \underset{\set{A}: P(\set{A})=1}{\max} \log \frac1{Q(\set{A})}, \\[-0.1cm]
\label{def:d1}
D_1(P\|Q) &= D(P\|Q), \\[-0.1cm]
\label{def:dinf}
D_{\infty}(P\|Q) &= \log \, \sup_{x \in \set{X}} \frac{P(x)}{Q(x)},
\end{align}
where $D(P\|Q)$ in the right side of \eqref{def:d1} is the relative entropy (a.k.a. Kullback-Leibler divergence).
\end{itemize}
\end{definition}

Throughout this paper, for $a \in \Reals$, $\lceil a \rceil$ denotes the ceiling of $a$ (i.e.,
the smallest integer not smaller than the real number $a$), and $\lfloor a \rfloor$ denotes
the flooring of $a$ (i.e., the greatest integer not greater than $a$).

\section{A Tight Lower Bound on the R\'{e}nyi Entropy}
\label{section:thm-RE}

We provide in this section a tight lower bound on the R\'{e}nyi entropy, of an arbitrary
order $\alpha > 0$, when the probability mass function of the discrete random variable
is defined on a finite set of cardinality
$n$, and the ratio of the maximal to minimal probability masses is upper bounded by an arbitrary
fixed value $\rho \in [1, \infty)$. In other words, we derive the largest possible gap between
the order-$\alpha$ R\'{e}nyi entropies of an equiprobable distribution and a non-equiprobable
distribution (defined on a finite set of the same cardinality) with a given value for the ratio
of the maximal to minimal probability masses. The basic tool used for the development of our
result in this section relies on the majorization theory. Our result strengthes the result in
\cite[Theorem~2]{CicaleseGV18} for the Shannon entropy, and it further provides a generalization
for the R\'{e}nyi entropy of an arbitrary order $\alpha > 0$ (recall that the Shannon entropy is
equal to the R\'{e}nyi entropy of order $\alpha=1$, see \eqref{eq: Shannon entropy}). Furthermore,
the approach for proving the main result in this section differs significantly from the proof in
\cite{CicaleseGV18} for the Shannon entropy. The main result in this section is a key result for
all what follows in this paper.

The following lemma is a restatement of \cite[Lemma~6]{CicaleseGV18}.

\begin{lemma}
\label{lemma: Lemma 6 from CicaleseGV18}
Let $P \in \set{P}_n(\rho)$ with $\rho \geq 1$ and an integer $n \geq 2$, and assume
without any loss of generality that the probability mass function $P$ is defined on the set
$\set{X} = \{1, \ldots, n\}$.
Let $Q \in \set{P}_n$ be defined on $\set{X}$ as follows:
\begin{align}
\label{eq:Q pmf}
Q(j) =
\begin{dcases}
\rho \, p_{\min}, & \quad \mbox{$j \in \{1, \ldots, i\}$,}\\
1-(n+i \rho - i-1) p_{\min}, & \quad \mbox{$j=i+1$,}\\
p_{\min}, &\quad \mbox{$j \in \{i+2, \ldots, n\}$}
\end{dcases}
\end{align}
where
\begin{align}
\label{eq:i index}
i := \left\lfloor \frac{1-np_{\min}}{(\rho-1) \, p_{\min}} \right\rfloor.
\end{align}
Then,
\begin{enumerate}[1)]
\item $Q \in \set{P}_n(\rho)$, and $Q(1) \geq Q(2) \geq \ldots \geq Q(n) > 0$;
\item $P \prec Q$.
\end{enumerate}
\end{lemma}

\begin{proof}
See \cite[p.~2236]{CicaleseGV18} (top of the second column).
\end{proof}

\vspace*{0.1cm}
\begin{lemma}
\label{lemma: min RE}
Let $\rho>1$, $\alpha > 0$, and $n \geq 2$ be an integer. For
\begin{align}
\label{eq: beta's range}
\beta \in \left[ \frac1{1+(n-1)\rho}, \, \frac1n \right] := \Gamma_\rho^{(n)}
\end{align}
let $Q_\beta \in \set{P}_n(\rho)$ be defined on $\set{X} = \{1, \ldots, n\}$ as follows:
\begin{align}
\label{eq:Q_beta pmf}
Q_\beta(j) =
\begin{dcases}
\rho \beta, & \quad \mbox{$j \in \{1, \ldots, i_\beta\}$,}\\
1-(n + i_\beta \, \rho - i_\beta - 1) \beta, & \quad \mbox{$j=i_\beta+1$,}\\
\beta, &\quad \mbox{$j \in \{i_\beta+2, \ldots, n\}$}
\end{dcases}
\end{align}
where
\begin{align}
\label{eq:i_beta index}
i_\beta := \left\lfloor \frac{1-n\beta}{(\rho-1) \beta} \right\rfloor.
\end{align}
Then, for every $\alpha > 0$,
\begin{align}
\label{eq: simplified opt.}
\min_{P \in \set{P}_n(\rho)} H_{\alpha}(P) = \min_{\beta \in \Gamma_\rho^{(n)}} H_{\alpha}(Q_\beta).
\end{align}
\end{lemma}

\begin{proof}
See Appendix~\ref{appendix: proof of min RE lemma}.
\end{proof}

\vspace*{0.1cm}
\begin{lemma}
\label{lemma: boundedness and monotonocity}
For $\rho > 1$ and $\alpha > 0$, let
\begin{align}
\label{eq: def. c}
c_{\alpha}^{(n)}(\rho) := \log n - \min_{P \in \set{P}_n(\rho)} H_{\alpha}(P), \quad n=2,3, \ldots
\end{align}
with $c_{\alpha}^{(1)}(\rho) := 0$. Then, for every $n \in \naturals$,
\begin{align}
& 0 \leq c_{\alpha}^{(n)}(\rho) \leq \log \rho, \label{eq: bounded} \\
& c_{\alpha}^{(n)}(\rho) \leq c_{\alpha}^{(2n)}(\rho), \label{eq: monotonicity}
\end{align}
and $c_{\alpha}^{(n)}(\rho)$ is monotonically increasing in $\alpha \in [0, \infty]$.
\end{lemma}

\begin{proof}
See Appendix~\ref{appendix: proof of boundedness and monotonocity}.
\end{proof}

\begin{lemma}
\label{lemma: c infinity}
For $\alpha > 0$ and $\rho > 1$, the limit
\begin{align}
\label{def: c inf.}
c_{\alpha}^{(\infty)}(\rho) & := \lim_{n \to \infty} c_{\alpha}^{(n)}(\rho)
\end{align}
exists, having the following properties:
\begin{enumerate}[a)]
\item \label{lemma: c infinity- 1)} If $\alpha \in (0,1) \cup (1, \infty)$, then
\begin{align}
\label{eq: c inf. - alpha neq 1}
c_{\alpha}^{(\infty)}(\rho) =
\frac1{\alpha-1} \, \log \left(1+\frac{1+\alpha \, (\rho-1) - \rho^\alpha}{(1-\alpha)(\rho-1)}\right)
- \frac{\alpha}{\alpha-1} \, \log \left(1+\frac{1+\alpha \, (\rho-1) - \rho^\alpha}{(1-\alpha)(\rho^\alpha-1)}\right),
\end{align}
and
\begin{align}
\label{eq: lim c inf. - alpha to inf.}
\lim_{\alpha \to \infty} c_{\alpha}^{(\infty)}(\rho) = \log \rho.
\end{align}
\item \label{lemma: c infinity- 2)} If $\alpha=1$, then
\begin{align}
\label{eq: c inf. for alpha=1}
c_1^{(\infty)}(\rho) = \lim_{\alpha \to 1} c_{\alpha}^{(\infty)}(\rho) = \frac{\rho \log \rho}{\rho-1}
- \log \left(\frac{\mathrm{e} \rho \log_{\mathrm{e}}\rho}{\rho-1} \right).
\end{align}
\item \label{lemma: c infinity- 3)} For all $\alpha > 0$,
\begin{align}
\label{eq: lim c inf. - rho to 1}
\lim_{\rho \downarrow 1} c_{\alpha}^{(\infty)}(\rho) = 0.
\end{align}
\item \label{lemma: c infinity- 4)} For every $n \in \naturals$, $\alpha > 0$ and $\rho \geq 1$,
\begin{align} \label{eq: tightened bound}
0 \leq c_{\alpha}^{(n)}(\rho) \leq c_{\alpha}^{(2n)}(\rho) \leq c_{\alpha}^{(\infty)}(\rho) \leq \log \rho.
\end{align}
\end{enumerate}
\end{lemma}

\begin{proof}
See Appendix~\ref{appendix: proof of lemma c_infinity}.
\end{proof}

In view of Lemmata~\ref{lemma: Lemma 6 from CicaleseGV18}--\ref{lemma: c infinity}, we obtain the following
main result in this section:
\begin{theorem}
\label{theorem: c}
Let $\alpha>0$, $\rho>1$, $n \geq 2$, and let $c_{\alpha}^{(n)}(\rho)$ in \eqref{eq: def. c} designate
the maximal gap between the order-$\alpha$ R\'{e}nyi entropies of equiprobable and arbitrary distributions
in $\set{P}_n(\rho)$. Then,
\begin{enumerate}[a)]
\item The non-negative sequence $\{c_{\alpha}^{(n)}(\rho)\}_{n=2}^{\infty}$ can be calculated
by the real-valued single-parameter optimization in the right side of \eqref{eq: simplified opt.}.
\item The asymptotic limit as $n \to \infty$, denoted
by $c_{\alpha}^{(\infty)}(\rho)$, admits the closed-form expressions in \eqref{eq: c inf. - alpha neq 1}
and \eqref{eq: c inf. for alpha=1}, and it satisfies the properties in
\eqref{eq: lim c inf. - alpha to inf.}, \eqref{eq: lim c inf. - rho to 1} and \eqref{eq: tightened bound}.
\end{enumerate}
\end{theorem}

\begin{remark}
Setting $\alpha = 2$ in Theorem~\ref{theorem: c} gives that, for all $P \in \set{P}_n(\rho)$
(with $\rho > 1$, and an integer $n \geq 2$),
\begin{align}
H_2(P) & \geq \log n - c_2^{(n)}(\rho) \label{eq-1: Simic} \\
& \geq \log n - c_2^{(\infty)}(\rho) \label{eq0: Simic} \\
&= \log \frac{4 \rho n}{(1+\rho)^2} \label{eq: Simic}
\end{align}
where \eqref{eq-1: Simic}, \eqref{eq0: Simic} and \eqref{eq: Simic} hold, respectively,
due to \eqref{eq: def. c}, \eqref{eq: tightened bound} and \eqref{eq: c inf. - alpha neq 1}.
This strengthens the result in \cite[Proposition~2]{Simic09} which gives
the same lower bound as in the right side of \eqref{eq: Simic} for $H(P)$ rather
than for $H_2(P)$ (recall that $H(P) \geq H_2(P)$).
\end{remark}

For a numerical illustration of Theorem~\ref{theorem: c}, Figure~\ref{figure:c_alpha_plot1}
provides a plot of $c_{\alpha}^{(\infty)}(\rho)$ in \eqref{eq: c inf. - alpha neq 1} and
\eqref{eq: c inf. for alpha=1} as a function of $\rho \geq 1$, confirming numerically
the properties in \eqref{eq: lim c inf. - alpha to inf.} and \eqref{eq: lim c inf. - rho to 1}.
Furthermore, Figure~\ref{figure:c_alpha_plot2} provides plots of $c_{\alpha}^{(n)}(\rho)$ in
\eqref{eq: def. c} as a function of $\alpha > 0$, for $\rho=2$ (left plot) and $\rho=256$
(right plot), with several values of $n \geq 2$; the calculation of the curves in these plots
relies on \eqref{eq: simplified opt.}, \eqref{eq: c inf. - alpha neq 1} and
\eqref{eq: c inf. for alpha=1}, and they illustrate the monotonicity and boundedness properties
in \eqref{eq: tightened bound}.

\begin{figure}[h]
\vspace*{-4.5cm}
\centerline{\includegraphics[width=10.2cm]{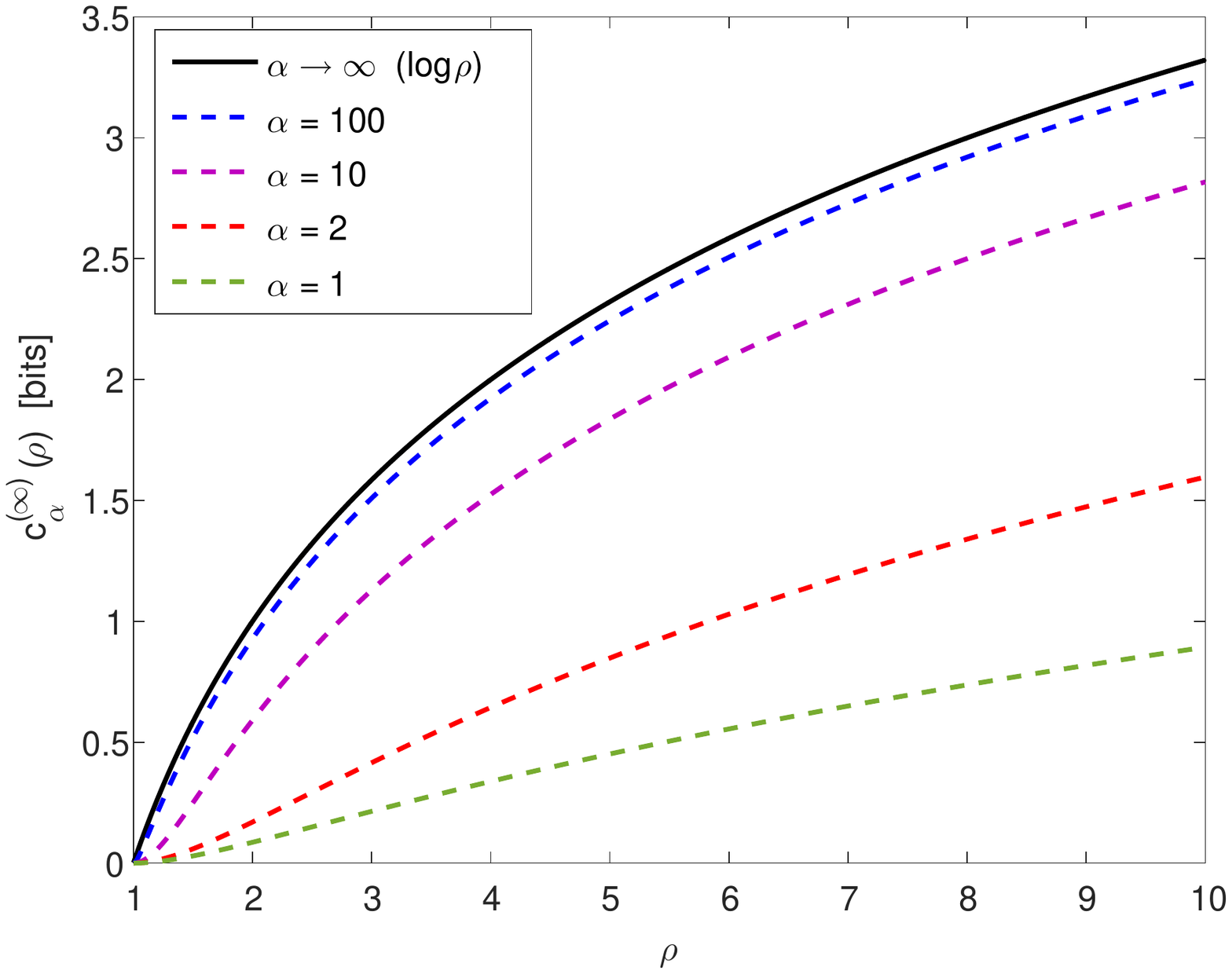}}
\vspace*{-3.7cm}
\caption{\label{figure:c_alpha_plot1}
A plot of $c_{\alpha}^{(\infty)}(\rho)$ in \eqref{eq: c inf. - alpha neq 1} and
\eqref{eq: c inf. for alpha=1} ($\log$ is on base~2) as a function of $\rho$, confirming numerically
the properties in \eqref{eq: lim c inf. - alpha to inf.} and \eqref{eq: lim c inf. - rho to 1}.}
\end{figure}

\begin{figure}[h]
\vspace*{-4.3cm}
\hspace*{-0.5cm}
\includegraphics[width=10.2cm]{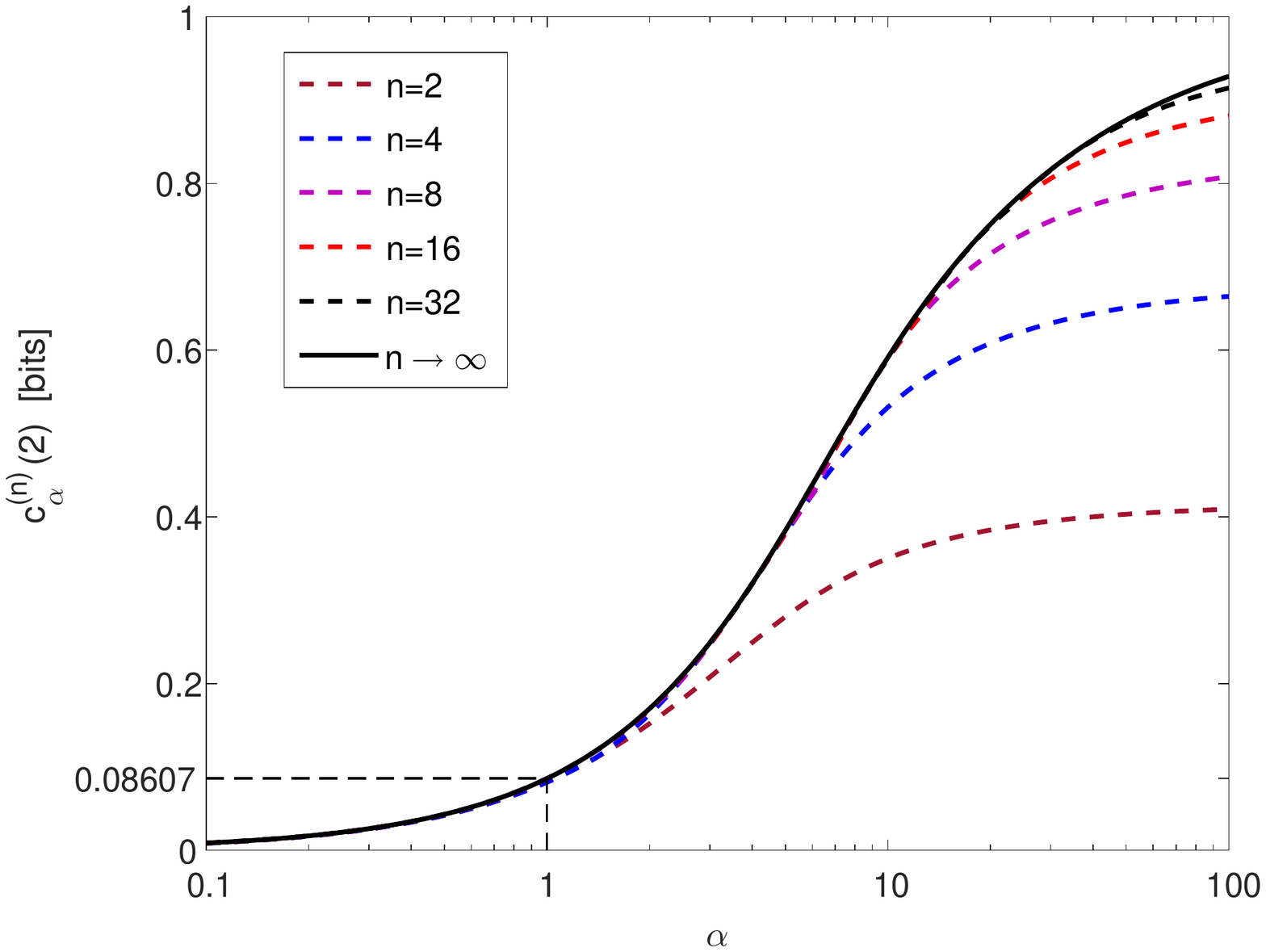} 
\hspace*{-1cm}
\includegraphics[width=10.2cm]{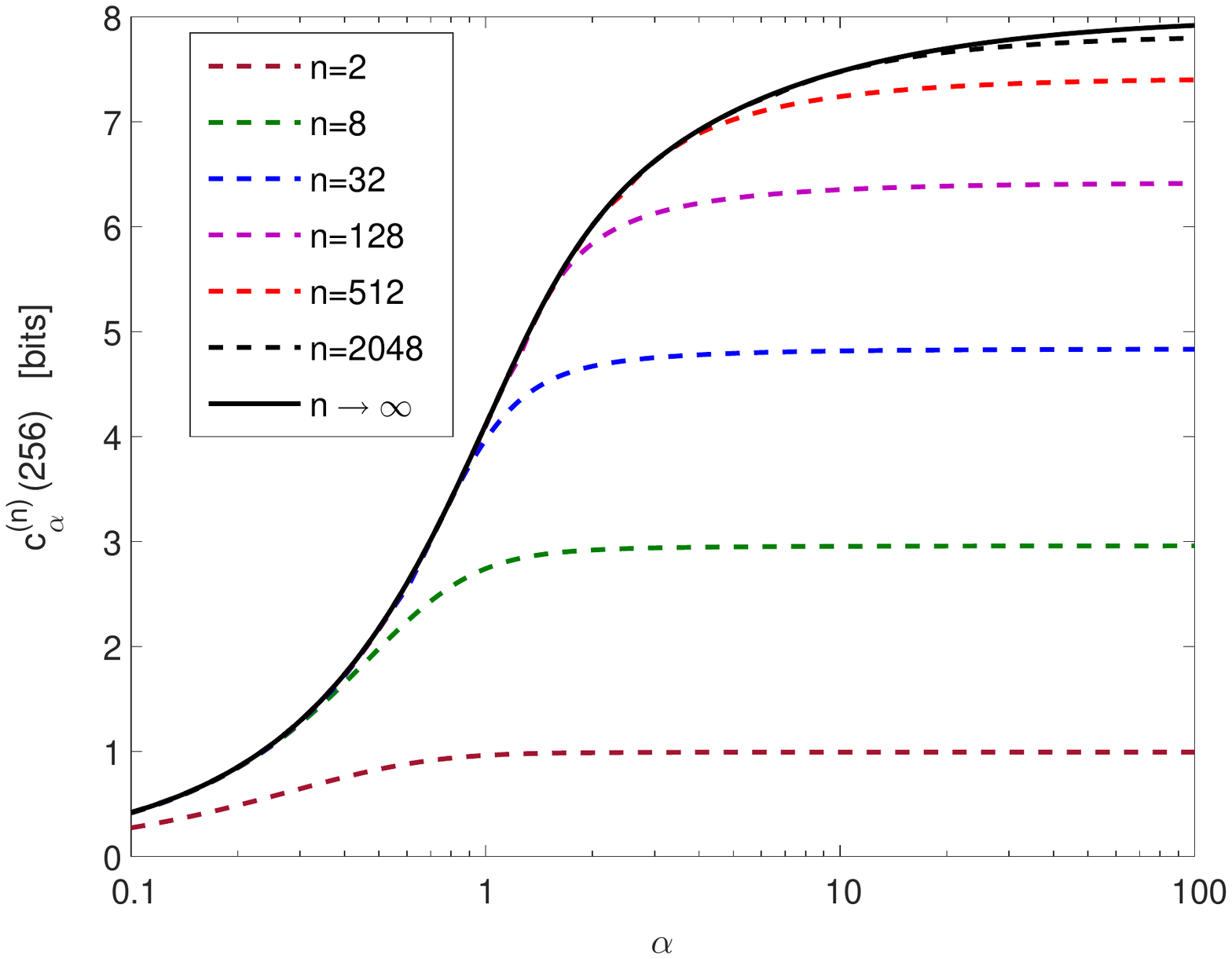}
\vspace*{-4.2cm}
\caption{\label{figure:c_alpha_plot2}
Plots of $c_{\alpha}^{(n)}(\rho)$ in \eqref{eq: def. c} ($\log$ is on base~2) as a function of $\alpha > 0$,
for $\rho=2$ (left plot) and $\rho=256$ (right plot), with several values of $n \geq 2$.}
\end{figure}

\begin{remark}
\label{remark: Th. 1 - entropy}
Theorem~\ref{theorem: c} strengthens the result in \cite[Theorem~2]{CicaleseGV18}
for the Shannon entropy (i.e., for $\alpha=1$), in addition to its generalization to
R\'{e}nyi entropies of arbitrary orders $\alpha>0$. This is because our lower bound
on the Shannon entropy is given by
\begin{align}
\label{eq: 20181112-a}
H(P) \geq \log n - c_1^{(n)}(\rho), \quad \forall \, P \in \set{P}_n(\rho),
\end{align}
whereas the looser bound in \cite{CicaleseGV18} is given by (see \cite[(7)]{CicaleseGV18}
and \eqref{eq: c inf. for alpha=1} here)
\begin{align}
\label{eq: 20181112-b}
H(P) \geq \log n - c_1^{(\infty)}(\rho), \quad \forall \, P \in \set{P}_n(\rho),
\end{align}
and we recall that $0 \leq c_1^{(n)}(\rho) \leq c_1^{(\infty)}(\rho)$ (see
\eqref{eq: tightened bound}). Figure~\ref{figure:improvement for Shannon entropy}
shows the improvement in the new lower bound \eqref{eq: 20181112-a} over \eqref{eq: 20181112-b}
by comparing $c_1^{(\infty)}(\rho)$ versus $c_1^{(n)}(\rho)$ for $\rho \in [1, 10^5]$
and with several values of $n$. It is reflected from
Figure~\ref{figure:improvement for Shannon entropy} that there is a very marginal
improvement in the lower bound on the Shannon entropy \eqref{eq: 20181112-a} over the
bound in \eqref{eq: 20181112-b} if $\rho \leq 30$ (even for small values of $n$), whereas
there is a significant improvement over the bound in \eqref{eq: 20181112-b} for large values
of $\rho$; by increasing the value of $n$, also the value of $\rho$
needs to be increased for observing an improvement of the lower bound in \eqref{eq: 20181112-a}
over \eqref{eq: 20181112-b} (see Figure~\ref{figure:improvement for Shannon entropy}).
\begin{figure}[h]
\vspace*{-4.5cm}
\centerline{\includegraphics[width=10.2cm]{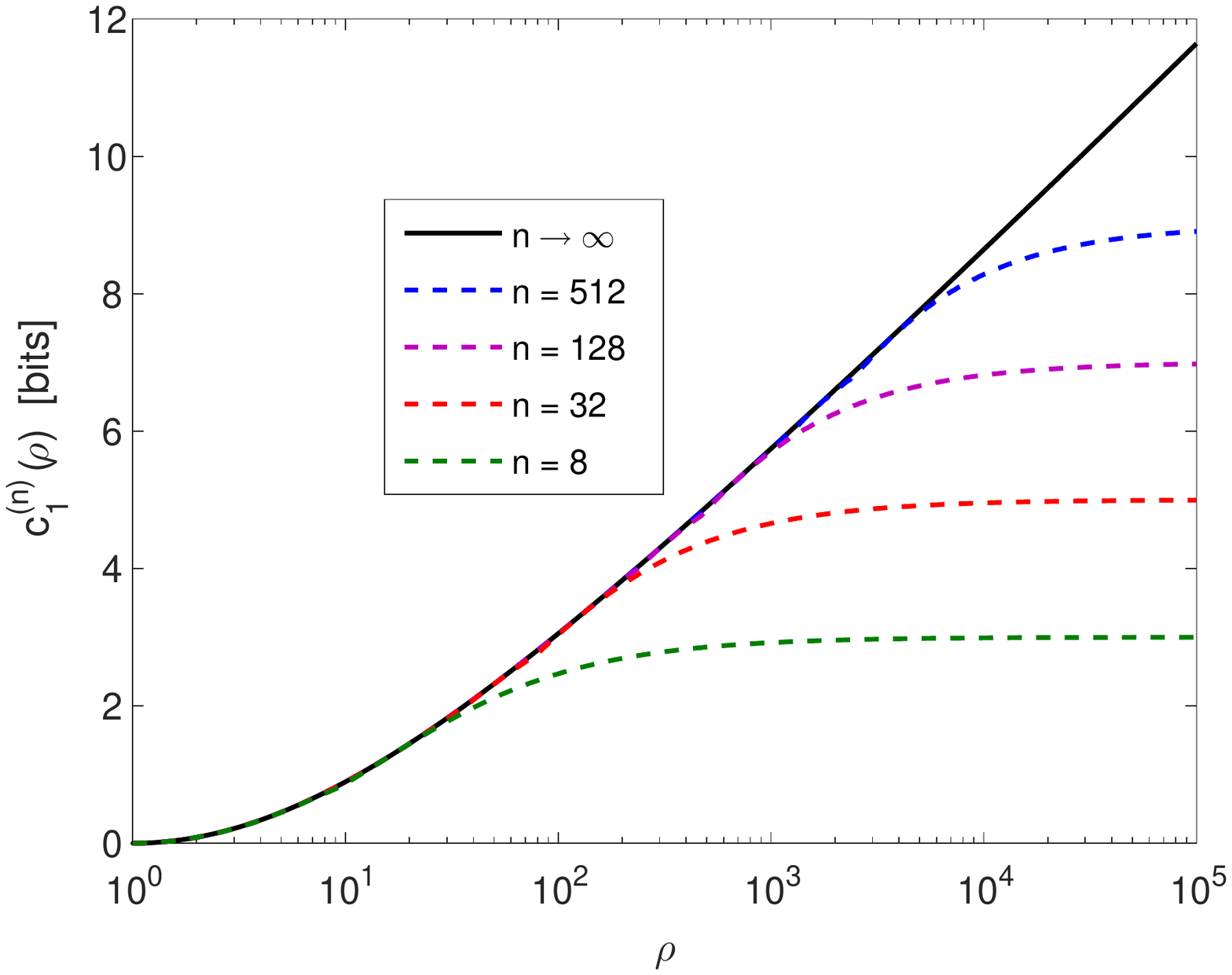}}
\vspace*{-3.7cm}
\caption{\label{figure:improvement for Shannon entropy}
A plot of $c_1^{(\infty)}(\rho)$ in \eqref{eq: c inf. for alpha=1}
versus $c_1^{(n)}(\rho)$ for finite $n$ ($n = 512, 128, 32$, and 8)
as a function of $\rho$.}
\end{figure}

An improvement of the bound in \eqref{eq: 20181112-a} over
\eqref{eq: 20181112-b} leads to a tightening of the upper bound in \cite[Theorem~4]{CicaleseGV18}
on the compression rate of Tunstall codes for discrete memoryless sources, which further tightens
the bound by Jelinek and Schneider in \cite[Eq.~(9)]{JelSc72}. More explicitly, in
view of \cite[Section~6]{CicaleseGV18}, an improved upper bound on the compression rate of these
variable-to-fixed lossless source codes is obtained by combining \cite[Eqs.~(36) and~(38)]{CicaleseGV18}
with a tightened lower bound on the entropy $H(W)$ of the leaves of the tree graph for Tunstall codes.
From \eqref{eq: 20181112-a}, the latter lower bound is given by $H(W) \geq \log_2 n - c_1^{(n)}(\rho)$
where $c_1^{(n)}(\rho)$ is expressed in bits, $\rho := \frac1{p_{\min}}$ is the reciprocal of the
minimal positive probability of the source symbols, and $n$ is the number of codewords (so, all codewords
are of length $\lceil \log_2 n \rceil$ bits). This yields a reduction in the upper bound on the
non-asymptotic compression rate $R$ of Tunstall codes from $\frac{\lceil \log_2 n \rceil \, H(X)}{\log_2 n - c_1^{(\infty)}(\rho)}$
(see \cite[Eq.~(40)]{CicaleseGV18} and \eqref{eq: c inf. for alpha=1}) to
$\frac{\lceil \log_2 n \rceil \, H(X)}{\log_2 n - c_1^{(n)}(\rho)}$ bits per source symbol where $H(X)$
denotes the source entropy (converging, in view of \eqref{eq: bounded}, to $H(X)$ as we let $n \rightarrow \infty$).
\end{remark}

\begin{remark}
Equality \eqref{eq: simplified opt.} with the minimizing probability mass function
of the form \eqref{eq:Q_beta pmf} holds, in general, by replacing the R\'{e}nyi entropy
with an arbitrary Schur-concave function (as it can be easily verified from the proof
of Lemma~\ref{lemma: min RE} in Appendix~\ref{appendix: proof of min RE lemma}). However,
the analysis leading to Lemmata~\ref{lemma: boundedness and monotonocity}--\ref{lemma: c infinity}
and Theorem~\ref{theorem: c} applies particularly to the R\'{e}nyi entropy.
\end{remark}

\section{Bounds on the R\'{e}nyi Entropy of a Function of a Discrete Random Variable}
\label{section: Function of RV}

This section relies on Theorem~\ref{theorem: c} and majorization for
extending \cite[Theorem~1]{CicaleseGV18}, which applies to the Shannon entropy, to
R\'{e}nyi entropies of any positive order. More explicitly, let $\alpha \in (0, \infty)$ and
\begin{itemize}
\item $\set{X}$ and $\set{Y}$ be finite sets of cardinalities $|\set{X}|=n$ and $|\set{Y}|=m$
with $n>m \geq 2$; without any loss of generality, let $\set{X} = \{1, \ldots, n\}$ and
$\set{Y} = \{1, \ldots, m\}$;
\item $X$ be a random variable taking values on $\set{X}$ with a probability mass function
$P_X \in \set{P}_n$;
\item $\set{F}_{n,m}$ be the set of deterministic functions $f \colon \set{X} \to \set{Y}$;
note that $f \in \set{F}_{n,m}$ is not one-to-one since $m<n$.
\end{itemize}

The main result in this section sharpens the inequality $H_{\alpha}\bigl(f(X)\bigr) \leq H_{\alpha}(X)$,
for every deterministic function $f \in \set{F}_{n,m}$ with $n>m \geq 2$ and $\alpha > 0$, by obtaining
non-trivial upper and lower bounds on $\underset{f \in \set{F}_{n,m}}{\max} H_{\alpha}\bigl( f(X) \bigr)$.
The calculation of the exact value of $\underset{f \in \set{F}_{n,m}}{\min} H_{\alpha}\bigl( f(X) \bigr)$
is much easier, and it is expressed in closed form by capitalizing on the Schur-concavity of the R\'enyi entropy.

The following main result extends \cite[Theorem~1]{CicaleseGV18} to R\'{e}nyi entropies of arbitrary positive orders.
\begin{theorem}
\label{thm:Huffman}
Let $X \in \{1, \ldots, n\}$ be a random variable which satisfies $P_X(1) \geq P_X(2) \geq \ldots \geq P_X(n)$.
\begin{enumerate}[a)]
\item
For $m \in \{2, \ldots, n-1\}$, if $P_X(1) < \frac1m$, let $\widetilde{X}_m$ be the equiprobable
random variable on $\{1, \ldots, m\}$; otherwise, if $P_X(1) \geq \frac1m$,
let $\widetilde{X}_m \in \{1, \ldots, m\}$ be a random variable with the probability mass function
\begin{align}
\label{eq:P_Xm}
P_{\widetilde{X}_m}(i) =
\begin{dcases}
P_X(i), & i \in \{1, \ldots, n^\ast\}, \\
\frac1{m-n^\ast} \sum_{j = n^\ast+1}^n P_X(j), & i \in \{n^\ast+1, \ldots, m\},
\end{dcases}
\end{align}
where $n^\ast$ is the maximal integer $i \in \{1, \ldots, m-1\}$ such that
\begin{align}
\label{eq: n ast}
P_X(i) \geq \frac1{m-i} \sum_{j=i+1}^n P_X(j).
\end{align}
Then, for every $\alpha > 0$,
\begin{align}
\label{eq:bounds on RE of f(X)}
\underset{f \in \set{F}_{n,m}}{\max} H_{\alpha}\bigl( f(X) \bigr)
\in \bigl[ H_\alpha(\widetilde{X}_m) - v(\alpha), \, H_\alpha(\widetilde{X}_m) \bigr],
\end{align}
where
\begin{align}
\label{eq:v}
v(\alpha) := c_\alpha^{(\infty)}(2) =
\begin{dcases}
\log \left( \frac{\alpha-1}{2^\alpha-2} \right) - \frac{\alpha}{\alpha-1} \,
\log \left(\frac{\alpha}{2^\alpha-1}\right), & \alpha \neq 1, \\[0.15cm]
\log \left(\frac{2}{e \, \ln 2}\right) \approx 0.08607 \, \text{bits}, & \alpha=1.
\end{dcases}
\end{align}
\item
There exists an explicit construction of a deterministic function
$f^\ast \in \set{F}_{n,m}$ such that
\begin{align}
\label{eq: construction}
H_{\alpha}\bigl( f^\ast(X) \bigr) \in \bigl[ H_\alpha(\widetilde{X}_m) - v(\alpha), \, H_\alpha(\widetilde{X}_m)\bigr]
\end{align}
where $f^\ast$ is independent of $\alpha$, and it is obtained by using
Huffman coding (as in \cite{CicaleseGV18} for $\alpha=1$).
\item
Let $\widetilde{Y}_m \in \{1, \ldots, m\}$ be a random variable with the probability mass function
\begin{align}
\label{eq:P_Ym}
P_{\widetilde{Y}_m}(i) =
\begin{dcases}
\sum_{k=1}^{n-m+1} P_X(k), & i=1, \\
P_X(n-m+i), & i \in \{2, \ldots, m\}.
\end{dcases}
\end{align}
Then, for every $\alpha > 0$,
\begin{align}
\label{eq:min H_alpha}
\underset{f \in \set{F}_{n,m}}{\min} H_{\alpha}\bigl( f(X) \bigr) = H_{\alpha}(\widetilde{Y}_m).
\end{align}
\end{enumerate}
\end{theorem}

\begin{remark}
Setting $\alpha=1$ specializes Theorem~\ref{thm:Huffman} to \cite[Theorem~1]{CicaleseGV18} (with regard to
the Shannon entropy). This point is further elaborated in Remark~\ref{remark: Th. 2}, after the proof of
Theorem~\ref{thm:Huffman}.
\end{remark}

\begin{remark}
Similarly to \cite[Lemma~1]{CicaleseGV18}, an exact solution of the maximization problem in the left side of
\eqref{eq:bounds on RE of f(X)} is strongly NP-hard \cite{GareyJ79}; this means that, unless
$\text{P}=\text{NP}$, there is no polynomial time algorithm which, for an arbitrarily small
$\varepsilon > 0$, computes an admissible deterministic function $f_{\varepsilon} \in \set{F}_{n,m}$
such that
\begin{align}
H_{\alpha}\bigl(f_{\varepsilon}(X)\bigr) \geq (1-\varepsilon) \underset{f \in \set{F}_{n,m}}{\max} H_{\alpha}\bigl( f(X) \bigr).
\end{align}
This motivates the derivation of the bounds in \eqref{eq:bounds on RE of f(X)}, and the simple
construction of a deterministic function $f^\ast \in \set{F}_{n,m}$ achieving \eqref{eq: construction}.
\end{remark}

A proof of Theorem~\ref{thm:Huffman} relies on the following lemmata.
\begin{lemma}
\label{lemma5}
Let $X \in \{1, \ldots, n\}$, $m<n$ and $\alpha > 0$. Then,
\begin{align}
\label{eq:lemma5}
\underset{Q \in \set{P}_m: \; P_X \prec Q}{\max} \; H_{\alpha}(Q) = H_{\alpha}(\widetilde{X}_m)
\end{align}
where the probability mass function of $\widetilde{X}_m$ is given in \eqref{eq:P_Xm}.
\end{lemma}
\begin{proof}
Since $P_X \prec P_{\widetilde{X}_m}$ (see \cite[Lemma~2]{CicaleseGV18}) with $P_{\widetilde{X}_m} \in \set{P}_m$,
and $P_{\widetilde{X}_m} \prec Q$ for all $Q \in \set{P}_m$ such that $P_X \prec Q$ (see \cite[Lemma~4]{CicaleseGV18}),
the result follows from the Schur-concavity of the R\'{e}nyi entropy.
\end{proof}

\begin{lemma}
\label{lemma6}
Let $X \in \{1, \ldots, n\}$, $\alpha > 0$, and $f \in \set{F}_{n,m}$ with $m<n$. Then,
\begin{align}
\label{eq:lemma6}
H_{\alpha}\bigl( f(X) \bigr) \leq H_{\alpha}(\widetilde{X}_m).
\end{align}
\end{lemma}
\begin{proof}
Since $f$ is a deterministic function in $\set{F}_{n,m}$ with $m<n$,
the probability mass function of $f(X)$ is an element in $\set{P}_m$
which majorizes $P_X$ (see \cite[Lemma~3]{CicaleseGV18}).
Inequality \eqref{eq:lemma6} then follows from Lemma~\ref{lemma5}.
\end{proof}

We are now ready to prove Theorem~\ref{thm:Huffman}.

\begin{proof}
In view of \eqref{eq:lemma6},
\begin{align}
\label{eq:th2-ub}
\underset{f \in \set{F}_{n,m}}{\max} H_{\alpha}\bigl(f(X)\bigr) \leq H_{\alpha}(\widetilde{X}_m).
\end{align}
We next construct a function $f^\ast \in \set{F}_{n,m}$ such that, for all $\alpha > 0$,
\begin{align}
H_{\alpha}\bigl(f^\ast(X)\bigr) & \geq \underset{Q \in \set{P}_m: \;
P_X \prec Q}{\max} \; H_{\alpha}(Q) - v(\alpha) \label{eq:th2-a} \\
& \geq \underset{f \in \set{F}_{n,m}}{\max} H_{\alpha}\bigl(f(X)\bigr) - v(\alpha) \label{eq:th2-b}
\end{align}
where the function $v \colon (0, \infty) \to (0, \infty)$ in the right side of \eqref{eq:th2-a} is
given in \eqref{eq:v}, and \eqref{eq:th2-b} holds due to \eqref{eq:lemma5} and \eqref{eq:th2-ub}.
The function $f^\ast$ in our proof coincides with the construction in \cite{CicaleseGV18},
and it is therefore independent of $\alpha$.

We first review and follow the concept of the proof of \cite[Lemma~5]{CicaleseGV18},
and we then deviate from the analysis there for proving our result. The idea behind
the proof of \cite[Lemma~5]{CicaleseGV18} relies on the following algorithm:
\begin{enumerate}[1)]
\item Start from the probability mass function $P_X \in \set{P}_n$ with $P_X(1) \geq \ldots \geq P_X(n)$;
\item Merge successively pairs of probability masses by applying the Huffman algorithm;
\item Stop the process in Step~2 when a probability mass function
$Q \in \set{P}_m$ is obtained (with $Q(1) \geq \ldots \geq Q(m)$);
\item Construct the deterministic function $f^\ast \in \set{F}_{n,m}$
by setting $f^\ast(k)=j \in \{1, \ldots, m\}$ for all probability masses $P_X(k)$,
with $k \in \{1, \ldots, n\}$, being merged in Steps~2--3 into the node of $Q(j)$.
\end{enumerate}
Let $i \in \{0, \ldots, m-1\}$ be the largest index such that $P_X(1)=Q(1), \ldots, P_X(i)=Q(i)$
(note that $i=0$ corresponds to the case where each node $Q(j)$, with $j \in \{1, \ldots, m\}$,
is constructed by merging at least two masses of the probability mass function $P_X$).
Then, according to \cite[p.~2225]{CicaleseGV18},
\begin{align}
\label{eq: factor 2}
Q(i+1) \leq 2 \, Q(m).
\end{align}
Let
\begin{align}
\label{eq:S}
S := \sum_{j=i+1}^m Q(j)
\end{align}
be the sum of the $m-i$ smallest masses of the probability mass function $Q$.
In view of \eqref{eq: factor 2}, the vector
\begin{align}
\label{eq: overline Q}
\overline{Q} := \left( \frac{Q(i+1)}{S}, \ldots, \frac{Q(m)}{S} \right)
\end{align}
represents a probability mass function where the ratio of its maximal to minimal masses
is upper bounded by~2.

At this point, our analysis deviates from \cite[p.~2225]{CicaleseGV18}.
Applying Theorem~\ref{theorem: c} to $\overline{Q}$ with $\rho=2$ gives
\begin{align}
\label{eq: key result - Th1}
H_{\alpha}(\overline{Q}) \geq \log(m-i) - c_{\alpha}^{(\infty)}(2)
\end{align}
with
\begin{align}
c_{\alpha}^{(\infty)}(2)
&= \frac1{\alpha-1} \log \left( 1 + \frac{1+\alpha-2^\alpha}{1-\alpha} \right)
- \frac{\alpha}{\alpha-1} \log \left(1 + \frac{1+\alpha-2^\alpha}{(1-\alpha)(2^\alpha-1)} \right) \label{eq:th2-c} \\
&= \log \left( \frac{\alpha-1}{2^\alpha-2} \right)
- \frac{\alpha}{\alpha-1} \log \left( \frac{\alpha}{2^\alpha-1} \right) \label{eq:th2-d} \\
&= v(\alpha) \label{eq:th2-e}
\end{align}
where \eqref{eq:th2-c} follows from \eqref{eq: c inf. - alpha neq 1}; \eqref{eq:th2-d}
is straightforward algebra, and \eqref{eq:th2-e} is the definition in \eqref{eq:v}.

For $\alpha \in (0,1) \cup (1, \infty)$, we get
\begin{align}
H_{\alpha}(Q)
&= \frac1{1-\alpha} \, \log \left( \sum_{j=1}^m Q^\alpha(j) \right) \label{eq:th2-f} \\
&= \frac1{1-\alpha} \, \log \left( \sum_{j=1}^i Q^\alpha(j)
+ \sum_{j=i+1}^m Q^\alpha(j) \right) \label{eq:th2-g} \\[0.1cm]
&= \frac1{1-\alpha} \, \log \left( \sum_{j=1}^i Q^\alpha(j) + S^\alpha \,
\exp\bigl( (1-\alpha) H_{\alpha}(\overline{Q}) \bigr) \right) \label{eq:th2-h} \\[0.1cm]
&\geq \frac1{1-\alpha} \, \log \left( \sum_{j=1}^i Q^\alpha(j) + S^\alpha \,
\exp\Bigl( (1-\alpha) \bigl( \log(m-i) - v(\alpha) \bigr) \Bigr) \right) \label{eq:th2-i} \\[0.1cm]
&= \frac1{1-\alpha} \, \log \left( \sum_{j=1}^i Q^\alpha(j) + S^\alpha \,
(m-i)^{1-\alpha} \exp\bigl( (\alpha-1) \, v(\alpha) \bigr) \right) \label{eq:th2-j}
\end{align}
where \eqref{eq:th2-g} holds since $i \in \{0, \ldots, m-1\}$;
\eqref{eq:th2-h} follows from \eqref{eq: Renyi entropy} and \eqref{eq: overline Q};
\eqref{eq:th2-i} holds by \eqref{eq: key result - Th1}--\eqref{eq:th2-e}.

In view of \eqref{eq:S}, let $Q^\ast \in \set{P}_m$ be the probability mass function
which is given by
\begin{align}
\label{eq: Q ast}
Q^\ast(j) =
\begin{dcases}
Q(j), & j = 1, \ldots, i\\
\frac{S}{m-i}, & j=i+1, \ldots, m.
\end{dcases}
\end{align}
From \eqref{eq:th2-f}--\eqref{eq: Q ast}, we get
\begin{align}
H_{\alpha}(Q) &\geq \frac1{1-\alpha} \, \log \left( \sum_{j=1}^i \bigl(Q^\ast(j)\bigr)^\alpha
+ \sum_{j=i+1}^m \bigl(Q^\ast(j)\bigr)^\alpha \; \exp\bigl((\alpha-1) \, v(\alpha) \bigr) \right) \label{eq:th2-h2} \\
&= \frac1{1-\alpha} \, \log \left( \sum_{j=1}^m \bigl(Q^\ast(j)\bigr)^\alpha
+ \sum_{j=i+1}^m \bigl(Q^\ast(j)\bigr)^\alpha \; \Bigl(\exp\bigl((\alpha-1) \, v(\alpha) \bigr) - 1 \Bigr) \right) \label{eq:th2-i2} \\[0.1cm]
&= H_{\alpha}(Q^\ast) +
\frac1{1-\alpha} \, \log \left(1+T \; \Bigl(\exp\bigl((\alpha-1) \, v(\alpha) \bigr) - 1 \Bigr) \right) \label{eq:th2-j2}
\end{align}
with
\begin{align}
\label{eq:T}
T := \frac{\overset{m}{\underset{j=i+1}{\sum}} \bigl(Q^\ast(j)\bigr)^\alpha}{\overset{m}{\underset{j=1}{\sum}}  \bigl(Q^\ast(j)\bigr)^\alpha} \in [0,1].
\end{align}
Since $T \in [0,1]$ and $v(\alpha)>0$ for $\alpha>0$, it can be verified from
\eqref{eq:th2-h2}--\eqref{eq:th2-j2} that for $\alpha \in (0,1) \cup (1, \infty)$
\begin{align}
\label{eq:th2-k}
H_{\alpha}(Q) \geq H_{\alpha}(Q^\ast) - v(\alpha).
\end{align}
The validity of \eqref{eq:th2-k} is extended to $\alpha=1$ by taking the limit $\alpha \to 1$
on both sides of this inequality, and due to the continuity of $v(\cdot)$ in \eqref{eq:v}
at $\alpha=1$. Applying the majorization result $Q^\ast \prec P_{\widetilde{X}_m}$ in \cite[(31)]{CicaleseGV18},
it follows from \eqref{eq:th2-k} and the Schur-concavity of the R\'{e}nyi entropy that, for
all $\alpha > 0$,
\begin{align}
H_{\alpha}(Q) \geq H_{\alpha}(Q^\ast) - v(\alpha) \geq H_{\alpha}(\widetilde{X}_m) - v(\alpha),
\end{align}
which together with \eqref{eq:th2-ub}, prove Items a) and b) of Theorem~\ref{thm:Huffman}
(note that, in view of the construction of the deterministic function $f^\ast \in \set{F}_{n,m}$
in Step~4 of the above algorithm, we get $H_{\alpha}\bigl(f^\ast(X)\bigr) = H_\alpha(Q)$).

We next prove Item c).
Equality~\eqref{eq:min H_alpha} is due to the Schur-concavity of the R\'{e}nyi entropy, and since we have
\begin{itemize}
\item
$f(X)$ is an {\em aggregation} of $X$, i.e., the probability mass function $Q \in \set{P}_m$ of $f(X)$ satisfies
$Q(j) = \underset{i \in I_j}{\sum} P_X(i)$ ($1\leq j \leq m$) where $I_1, \ldots, I_m$ partition
$\{1, \ldots, n\}$ into $m$ disjoint subsets as follows:
\begin{align}
I_j := \{i \in \{1, \ldots, n\}: f(i)=j\}, \quad j=1, \ldots, m;
\end{align}
\item
By the assumption $P_X(1) \geq P_X(2) \geq \ldots \geq P_X(n)$, it follows
that $Q \prec P_{\widetilde{Y}_m}$ for every such $Q \in \set{P}_m$;
\item
From \eqref{eq:P_Ym}, $\widetilde{Y}_m = \widetilde{f}(X)$ where the function $\widetilde{f} \in \set{F}_{n,m}$ is given by
$\widetilde{f}(k):=1$ for all $k \in \{1, \ldots, n-m+1\}$, and $\widetilde{f}(n-m+i):=i$ for all $i \in \{2, \ldots, m\}$.
Hence, $P_{\widetilde{Y}_m}$ is an element in the set of the probability mass functions of $f(X)$ with $f \in \set{F}_{n,m}$
which majorizes every other element from this set.
\end{itemize}
\end{proof}

\begin{remark}
The solid line in the left plot of Figure~\ref{figure:c_alpha_plot2} depicts
$v(\alpha) := c_{\alpha}^{(\infty)}(2)$ in \eqref{eq:v} for $\alpha>0$. In view
of Lemma~\ref{lemma: c infinity}, and by the definition in \eqref{eq:v}, the
function $v \colon (0, \infty) \to (0, \infty)$ is indeed monotonically increasing
and continuous.
\end{remark}

\begin{remark}
\label{remark: Th. 2}
Inequality \eqref{eq: factor 2} leads to the application of Theorem~\ref{theorem: c} with
$\rho=2$ (see \eqref{eq: key result - Th1}).
In the derivation of Theorem~\ref{thm:Huffman}, we refer to $v(\alpha) := c_{\alpha}^{(\infty)}(2)$
(see \eqref{eq:th2-c}--\eqref{eq:th2-e}) rather than referring to $c_{\alpha}^{(n)}(2)$ (although,
from \eqref{eq: tightened bound}, we have $0 \leq c_{\alpha}^{(n)}(2) \leq v(\alpha)$ for all $\alpha>0$).
We do so since, for $n \geq 16$, the difference between the curves of $c_{\alpha}^{(n)}(2)$
(as a function of $\alpha>0$) and the curve of $c_{\alpha}^{(\infty)}(2)$ is marginal (see
the dashed and solid lines in the left plot of Figure~\ref{figure:c_alpha_plot2}), and also
because the function $v$ in \eqref{eq:v} is expressed in a closed form whereas $c_{\alpha}^{(n)}(2)$
is subject to numerical optimization for finite $n$ (see \eqref{eq: simplified opt.} and \eqref{eq: def. c}).
For this reason, Theorem~\ref{thm:Huffman} coincides with the result in \cite[Theorem~1]{CicaleseGV18}
for the Shannon entropy (i.e., for $\alpha=1$) while providing a generalization of the latter result
for R\'{e}nyi entropies of arbitrary positive orders $\alpha$.
Theorem~\ref{theorem: c}, however, both strengthens the bounds in \cite[Theorem~2]{CicaleseGV18}
for the Shannon entropy with finite cardinality $n$ (see Remark~\ref{remark: Th. 1 - entropy}),
and it also generalizes these bounds to R\'{e}nyi entropies of all positive orders.
\end{remark}

\begin{remark}
The minimizing probability mass function in \eqref{eq:P_Ym} to the optimization problem \eqref{eq:min H_alpha},
and the maximizing probability mass function in \eqref{eq:P_Xm} to the optimization problem \eqref{eq:lemma5}
are in general valid when the R\'{e}nyi entropy of a positive order is replaced by an arbitrary Schur-concave
function. However, the main results in \eqref{eq:bounds on RE of f(X)}--\eqref{eq: construction} hold particularly
for the R\'{e}nyi entropy.
\end{remark}

\begin{remark}
\label{remark: notation in Theorem 2}
Theorem~\ref{thm:Huffman} makes use of the random variables denoted by $\widetilde{X}_m$ and $\widetilde{Y}_m$,
rather than (more simply) $X_m$ and $Y_m$ respectively, because Section~\ref{section: IT Applications}
considers i.i.d. samples $\{X_i\}_{i=1}^k$ and $\{Y_i\}_{i=1}^k$ with $X_i \sim P_X$ and $Y_i \sim P_Y$;
note, however, that the probability mass functions of $\widetilde{X}_m$ and $\widetilde{Y}_m$ are different
from $P_X$ and $P_Y$, respectively, and for that reason we make use of tilted symbols in the left sides of
\eqref{eq:P_Xm} and \eqref{eq:P_Ym}.
\end{remark}

\section{Information-Theoretic Applications: Non-Asymptotic Bounds For Lossless Compression and Guessing}
\label{section: IT Applications}

Theorem~\ref{thm:Huffman} is applied in this section to derive non-asymptotic
bounds for lossless compression of discrete memoryless sources, and guessing moments.
Each of the two subsections starts with a short background for making the presentation
self contained.

\subsection{Guessing}
\label{subsection: guessing}

\subsubsection{Background}
\label{background-guessing}
The problem of guessing discrete random variables has various theoretical and
operational aspects in information theory (see \cite{Arikan96,ArikanM98-1,
ArikanM98-2,Boztas97,BracherHL15,BurinS_IT18,ChDu13,HanawalS11,HanawalS11b,Huleihel17,
Massey94,McElieceYu95,MerhavAr99,PfisterSu04,SantisGV01,SasonV18b,Sundaresan07,
Sundaresan07b,Yona17}). The central object of interest is the distribution of
the number of guesses required to identify a realization of a random variable $X$, taking values
on a finite or countably infinite set $\set{X} = \{1, \ldots, |\set{X}|\}$, by successively
asking questions of the form ``Is $X$ equal to $x$?'' until the value of $X$ is guessed
correctly. A {\em guessing function} is a one-to-one function $g \colon \set{X} \to \set{X}$,
which can be viewed as a permutation of the elements of $\set{X}$ in the order in which they
are guessed. The required number of guesses is therefore equal to $g(x)$ when $X=x$ with
$x \in \set{X}$.

Lower and upper bounds on the minimal expected number of required guesses for correctly
identifying the realization of $X$, expressed as a function of the Shannon entropy $H(X)$,
have been respectively derived by Massey \cite{Massey94} and by McEliece and Yu
\cite{McElieceYu95}, followed by a derivation of improved upper and lower bounds by De
Santis {\em et al.} \cite{SantisGV01}. More generally, given a probability mass function $P_X$ on $\set{X}$,
it is of interest to minimize the generalized guessing moment
$\expectation[g^{\rho}(X)] = \underset{x \in \set{X}}{\sum} P_X(x) g^{\rho}(x)$ for $\rho > 0$.
For an arbitrary positive $\rho$, the $\rho$-th moment of the number
of guesses is minimized by selecting the guessing function to be a {\em ranking function}
$g_X$, for which $g_X(x)=\ell$ if $P_X(x)$ is the $\ell$-th largest mass \cite{Massey94}.
Although the tie breaking affects the choice of $g_X$, the distribution of
$g_X(X)$ does not depend on how ties are resolved. Not only does this strategy minimize
the average number of guesses, but it also minimizes the $\rho$-th moment of the number
of guesses for every $\rho > 0$.
Upper and lower bounds on the $\rho$-th moment of ranking functions, expressed
in terms of the R\'{e}nyi entropies,
were derived by Arikan \cite{Arikan96}, Bozta\c{s} \cite{Boztas97}, followed by recent
improvements in the non-asymptotic regime by Sason and Verd\'{u} \cite{SasonV18b}.
Although if $|\set{X}|$ is small, it is straightforward to evaluate numerically the guessing
moments, the benefit of bounds expressed in terms of R\'{e}nyi entropies is particularly
relevant when dealing with a random vector $X^k=(X_1, \ldots, X_k)$ whose letters belong
to a finite alphabet $\set{X}$; computing all the probabilities of the mass function $P_{X^k}$
over the set $\set{X}^k$, and then sorting them in decreasing order for the calculation of
the $\rho$-th moment of the optimal guessing function for the elements of $\set{X}^k$ becomes
infeasible even for moderate values of $k$. In contrast, regardless of the value of $k$,
bounds on guessing moments which depend on the R\'{e}nyi entropy are readily computable if
for example $\{X_i\}_{i=1}^k$ are independent; in which case, the R\'{e}nyi entropy of the
vector is equal to the sum of the R\'{e}nyi entropies of its components.
Arikan's bounds in \cite{Arikan96} are asymptotically tight for random vectors of length
$k$ as $k \to \infty$, thus providing the correct exponential growth rate of the guessing
moments for sufficiently large $k$.

\subsubsection{Analysis}
\label{analysis-guessing}
We next analyze the following setup of guessing. Let $\{X_i\}_{i=1}^k$ be i.i.d.
random variables where $X_1 \sim P_X$ takes values on a finite set $\set{X}$ with $|\set{X}|=n$.
In order to cluster the data \cite{Gan2007} (see also \cite[Section~3.A]{CicaleseGV18}
and references therein), suppose that each $X_i$ is mapped to $Y_i = f(X_i)$
where $f \in \set{F}_{n,m}$ is an arbitrary deterministic function (independent of the index~$i$)
with $m<n$. Consequently, $\{Y_i\}_{i=1}^k$ are i.i.d., and each $Y_i$ takes values on a finite set
$\set{Y}$ with $|\set{Y}|=m<|\set{X}|$.

Let $g_{X^k} \colon \set{X}^k \to \{1, \ldots, n^k\}$ and $g_{Y^k} \colon \set{Y}^k \to \{1, \ldots, m^k\}$
be, respectively, the ranking functions of the random vectors $X^k = (X_1, \ldots, X_k)$ and
$Y^k = (Y_1, \ldots, Y_k)$ by sorting in separate decreasing orders the probabilities
$P_{X^k}(x^k) = \prod_{i=1}^k P_X(x_i)$ for $x^k \in \set{X}^k$, and
$P_{Y^k}(y^k) = \prod_{i=1}^k P_Y(y_i)$ for $y^k \in \set{Y}^k$
where ties in both cases are resolved arbitrarily. In view of Arikan's bounds on the
$\rho$-th moment of ranking functions (see \cite[Theorem~1]{Arikan96} for the lower bound, and
\cite[Proposition~4]{Arikan96} for the upper bound), since $|\set{X}^k| = n^k$ and $|\set{Y}^k| = m^k$,
the following bounds hold for all $\rho > 0$:
\begin{align}
\rho H_{\frac1{1+\rho}}(X) - \frac{\rho \log(1+k \ln n)}{k} \leq \frac1k \, \log \expectation\bigl[ g_{X^k}^\rho(X^k) \bigr]
\leq \rho H_{\frac1{1+\rho}}(X), \label{eq:20181011-a} \\[0.2cm]
\rho H_{\frac1{1+\rho}}(Y) - \frac{\rho \log(1+k \ln m)}{k} \leq \frac1k \, \log \expectation\bigl[ g_{Y^k}^\rho(Y^k) \bigr]
\leq \rho H_{\frac1{1+\rho}}(Y). \label{eq:20181011-b}
\end{align}
In the following, we rely on Theorem~\ref{thm:Huffman} and the bounds in \eqref{eq:20181011-a}
and \eqref{eq:20181011-b} to obtain bounds on the exponential reduction of
the $\rho$-th moment of the ranking function of $X^k$ as a result of its mapping to $Y^k$.
First, the combination of \eqref{eq:20181011-a} and \eqref{eq:20181011-b} yields
\begin{align}
& \rho \left[ H_{\frac1{1+\rho}}(X) - H_{\frac1{1+\rho}}(Y) \right] - \frac{\rho \log(1+k \ln n)}{k} \nonumber \\[0.1cm]
& \leq \frac1k \log \frac{\expectation\bigl[ g_{X^k}^\rho(X^k) \bigr]}{\expectation\bigl[ g_{Y^k}^\rho(Y^k) \bigr]} \label{eq:20181011-c} \\[0.1cm]
& \leq \rho \left[ H_{\frac1{1+\rho}}(X) - H_{\frac1{1+\rho}}(Y) \right] + \frac{\rho \log(1+k \ln m)}{k}. \label{eq:20181011-d}
\end{align}
In view of Theorem~\ref{thm:Huffman}-a) and \eqref{eq:20181011-c}, it follows that for an arbitrary $f \in \set{F}_{n,m}$ and $\rho > 0$
\begin{align}
\label{eq:20181011-e}
\frac1k \log \frac{\expectation\bigl[ g_{X^k}^\rho(X^k) \bigr]}{\expectation\bigl[ g_{Y^k}^\rho(Y^k) \bigr]}
\geq \rho \left[ H_{\frac1{1+\rho}}(X) - H_{\frac1{1+\rho}}(\widetilde{X}_m) \right] - \frac{\rho \log(1+k \ln n)}{k}
\end{align}
where $\widetilde{X}_m$ is a random variable whose probability mass function is given in \eqref{eq:P_Xm}. Note that
\begin{align}
\label{eq:20181011-f}
H_{\frac1{1+\rho}}(\widetilde{X}_m) \leq H_{\frac1{1+\rho}}(X), \quad \frac{\rho \log(1+k \ln n)}{k} \underset{k \to \infty}{\longrightarrow} 0
\end{align}
where the first inequality in \eqref{eq:20181011-f} holds since $P_X \prec P_{\widetilde{X}_m}$ (see Lemma~\ref{lemma5})
and the R\'{e}nyi entropy is Schur-concave.

By the explicit construction of the function $f^\ast \in \set{F}_{n,m}$ according to the algorithm in
Steps~1--4 in the proof of Theorem~\ref{thm:Huffman} (based on the Huffman procedure), by setting
$Y_i := f^\ast(X_i)$ for every $i \in \{1, \ldots, k\}$, it follows from \eqref{eq: construction} and
\eqref{eq:20181011-d} that for all $\rho > 0$

\begin{align}
\label{eq:20181011-g}
\frac1k \log \frac{\expectation\bigl[ g_{X^k}^\rho(X^k) \bigr]}{\expectation\bigl[ g_{Y^k}^\rho(Y^k) \bigr]}
\leq \rho \left[ H_{\frac1{1+\rho}}(X) - H_{\frac1{1+\rho}}(\widetilde{X}_m) + v\hspace*{-0.07cm}\left(\frac1{1+\rho}\right) \right]
+ \frac{\rho \log(1+k \ln m)}{k}
\end{align}
where the monotonically increasing function $v \colon (0, \infty) \to (0, \infty)$ is given in \eqref{eq:v}, and it
is depicted by the solid line in the left plot of Figure~\ref{figure:c_alpha_plot2}. In view of \eqref{eq:v}, it can be shown that
the linear approximation $v(\alpha) \approx v(1) \, \alpha$ is excellent for all $\alpha \in [0,1]$, and therefore
for all $\rho > 0$
\begin{align}
\label{eq:20181011-i}
v\hspace*{-0.07cm}\left(\frac1{1+\rho}\right) \approx \frac{0.08607}{1+\rho} \; \text{bits}.
\end{align}
Hence, for sufficiently large value of $k$, the gap between the lower and upper bounds in
\eqref{eq:20181011-e} and \eqref{eq:20181011-g} is marginal, being approximately equal to
$\frac{0.08607 \, \rho}{1+\rho} \; \text{bits}$ for all $\rho > 0$.

The following theorem summarizes our result in this section.
\begin{theorem}
\label{thm: guessing}
Let
\begin{itemize}
\item
$\{X_i\}_{i=1}^k$ be i.i.d. with $X_1 \sim P_X$ taking values on a set
$\set{X}$ with $|\set{X}|=n$;
\item
$Y_i = f(X_i)$, for every $i \in \{1, \ldots, k\}$, where
$f \in \set{F}_{n,m}$ is a deterministic function with $m<n$;
\item
$g_{X^k} \colon \set{X}^k \to \{1, \ldots, n^k\}$ and $g_{Y^k} \colon \set{Y}^k \to \{1, \ldots, m^k\}$
be, respectively, ranking functions of the random vectors $X^k = (X_1, \ldots, X_k)$ and $Y^k = (Y_1, \ldots, Y_k)$.
\end{itemize}
Then, for every $\rho > 0$,
\begin{enumerate}[a)]
\item
The lower bound in \eqref{eq:20181011-e} holds for every deterministic function $f \in \set{F}_{n,m}$;
\item
The upper bound in \eqref{eq:20181011-g} holds for the specific $f^\ast \in \set{F}_{n,m}$, whose construction
relies on the Huffman algorithm (see Steps~1--4 of the procedure in the proof of Theorem~\ref{thm:Huffman});
\item
The gap between these bounds, for $f = f^\ast$ and sufficiently
large $k$, is at most $\rho \, v\hspace*{-0.07cm}\left(\frac1{1+\rho}\right) \approx \frac{0.08607 \, \rho}{1+\rho} \; \text{bits}.$
\end{enumerate}
\end{theorem}

\vspace*{0.1cm}
\subsubsection{Numerical Result}
\label{numerical results-guessing}
The following simple example illustrates the tightness of the achievable upper bound
and the universal lower bound in Theorem~\ref{thm: guessing}, especially for sufficiently
long sequences.
\begin{example}
\label{example: guessing}
\begin{figure}[h]
\vspace*{-4.2cm}
\hspace*{-1cm}
\includegraphics[width=10cm]{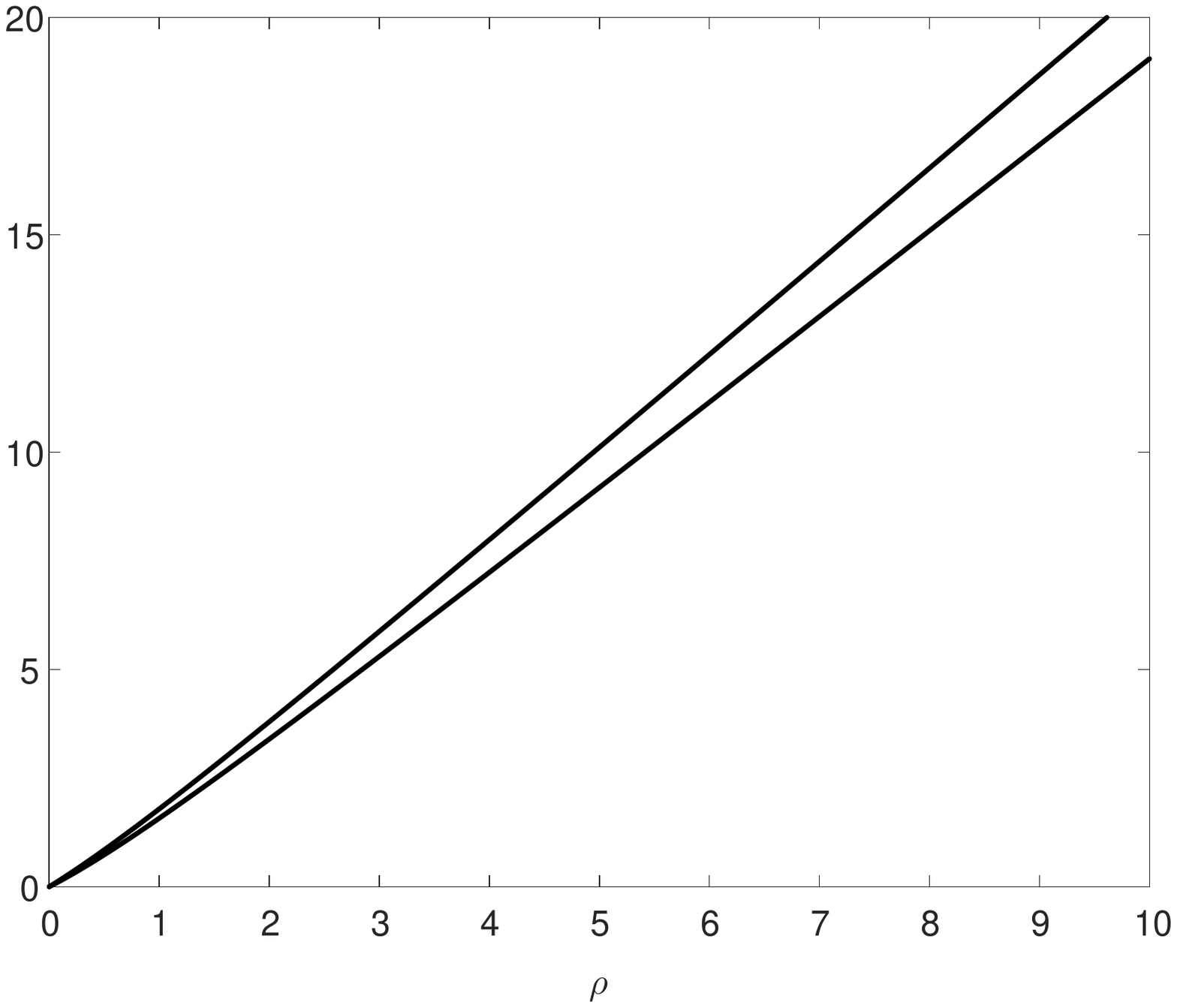}
\hspace*{-1.2cm}
\includegraphics[width=10cm]{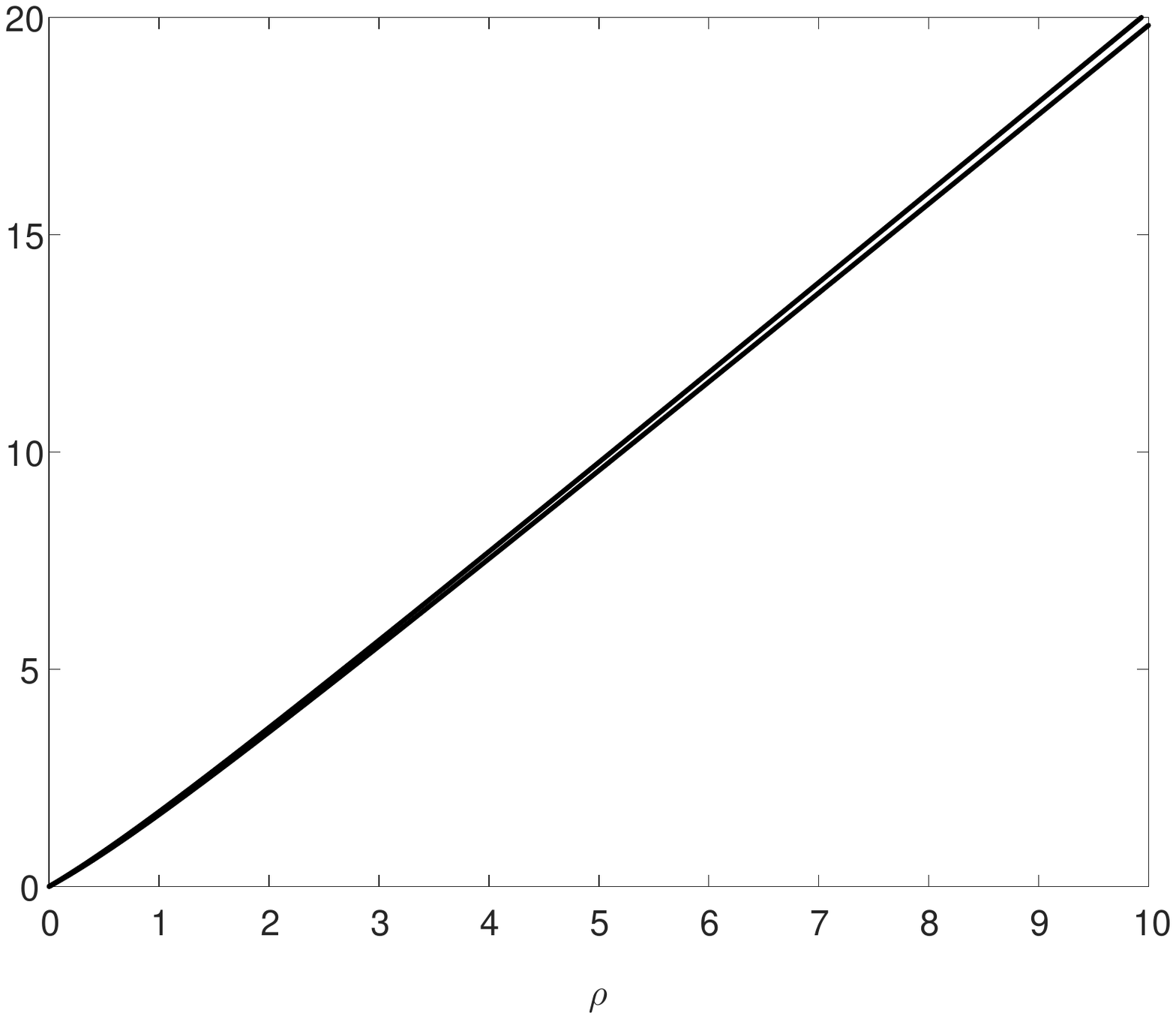}
\vspace*{-4cm}
\caption{\label{figure:guessing}
Plots of the upper and lower bounds on
$\frac1k \log_2 \frac{\expectation\bigl[ g_{X^k}^\rho(X^k) \bigr]}{\expectation\bigl[ g_{Y^k}^\rho(Y^k) \bigr]}$
in Theorem~\ref{thm: guessing}, as a function of $\rho > 0$, for random vectors of
length $k=100$ (left plot) or $k=1000$ (right plot) in the setting of
Example~\ref{example: guessing}. Each plot shows the universal lower bound for
an arbitrary deterministic $f \in \set{F}_{128, \, 16}$, and the achievable
upper bound with the construction of the deterministic function
$f=f^\ast \in \set{F}_{128, \, 16}$ (based on the Huffman algorithm) in
Theorem~\ref{thm: guessing} (see, respectively, \eqref{eq:20181011-e} and
\eqref{eq:20181011-g}).}
\end{figure}
Let $X$ be geometrically distributed restricted to $\{1, \ldots , n\}$ with
the probability mass function
\begin{align} \label{eq: geometric dist.}
P_X(j) = \frac{(1-a) \, a^{j-1}}{1-a^n}, \quad j \in \{1, \ldots, n\}
\end{align}
where $a = \tfrac{24}{25}$ and $n = 128$. Assume that $X_1, \ldots, X_k$
are i.i.d. with $X_1 \sim P_X$, and let $Y_i = f(X_i)$ with a deterministic
function $f \in \set{F}_{n,m}$ with $n =128$ and $m=16$. We compare the
upper and lower bounds in Theorem~\ref{thm: guessing} for the two cases
where the sequence $X^k = (X_1, \ldots, X_k)$ is of length $k=100$ or $k=1000$.
The lower bound in \eqref{eq:20181011-e} holds for an arbitrary deterministic
$f \in \set{F}_{n,m}$, and the achievable upper bound in \eqref{eq:20181011-g}
holds for the construction of the deterministic function
$f=f^\ast \in \set{F}_{n,m}$ (based on the Huffman algorithm) in
Theorem~\ref{thm: guessing}.

Numerical results are shown in Figure~\ref{figure:guessing}, providing plots
of the upper and lower bounds on
$\frac1k \log_2 \frac{\expectation\bigl[ g_{X^k}^\rho(X^k) \bigr]}{\expectation\bigl[ g_{Y^k}^\rho(Y^k) \bigr]}$
in Theorem~\ref{thm: guessing}, and illustrating the improved tightness of these
bounds when the value of $k$ is increased from 100 (left plot) to 1000 (right plot).
From Theorem~\ref{thm: guessing}-c), for sufficiently large $k$, the gap between
the upper and lower bounds is less than 0.08607~bits (for all $\rho>0$); this is
consistent with the right plot of Figure~\ref{figure:guessing} where $k=1000$.
\end{example}

\subsection{Lossless Source Coding}
\label{subsubsection: lossless source coding}

\subsubsection{Background}
\label{background-Campbell}
For uniquely-decodable (UD) lossless source coding, Campbell (\cite{Campbell65, Campbell66})
proposed the cumulant generating function of the codeword lengths as a generalization
to the frequently used design criterion of average code length. Campbell's
motivation in \cite{Campbell65} was to control the contribution of the longer codewords
via a free parameter in the cumulant generating function: if the value of this parameter
tends to zero, then the resulting design criterion becomes the average
code length per source symbol; on the other hand, by increasing the value of the free parameter, the penalty for longer
codewords is more severe, and the resulting code optimization yields a reduction in the fluctuations
of the codeword lengths.

We introduce the coding theorem by Campbell \cite{Campbell65} for lossless compression
of a discrete memoryless source (DMS) with UD codes, which serves for our analysis
jointly with Theorem~\ref{thm:Huffman}.
\begin{theorem}[Campbell 1965, \cite{Campbell65}]
\label{theorem: Campbell}
Consider a DMS which emits symbols with a probability
mass function $P_X$ defined on a (finite or countably infinite) set $\set{X}$.
Consider a UD fixed-to-variable source code operating on
source sequences of $k$ symbols with an alphabet of the codewords of size $D$.
Let $\ell(x^k)$ be the length of the codeword which corresponds to the source
sequence $x^k := (x_1, \ldots, x_k) \in \set{X}^k$. Consider the scaled
{\em cumulant generating function} of the codeword lengths\footnote{The
term {\em scaled} cumulant generating function is used in view of \cite[Remark~20]{SasonV18b}.}
\begin{align}
\label{eq: cumulant generating function}
\Lambda_k(\rho) := \frac1{k} \, \log_D \left( \, \sum_{x^k \in \set{X}^k}
P_{X^k}(x^k) \, D^{\rho \, \ell(x^k)} \right), \quad \rho > 0
\end{align}
where
\begin{align}
\label{eq: pmf}
P_{X^k}(x^k) = \prod_{i=1}^k P_X(x_i), \quad \forall \, x^k \in \set{X}^k.
\end{align}
Then, for every $\rho > 0$, the following hold:
\begin{enumerate}[a)]
\item Converse result:
\begin{align}
\label{eq: Campbell's converse result}
\frac{\Lambda_k(\rho)}{\rho} \geq \frac{1}{\log D} \; H_{\frac1{1+\rho}}(X).
\end{align}
\item Achievability result:
there exists a UD source code, for which
\begin{align}
\label{eq: Campbell's achievability result}
\frac{\Lambda_k(\rho)}{\rho} \leq \frac{1}{\log D} \; H_{\frac1{1+\rho}}(X) + \frac{1}{k}.
\end{align}
\end{enumerate}
\end{theorem}

The bounds in Theorem~\ref{theorem: Campbell}, expressed in terms of the R\'{e}nyi entropy,
imply that for sufficiently long source sequences, it is possible to make the scaled
cumulant generating function of the codeword lengths approach the R\'{e}nyi entropy as closely
as desired by a proper fixed-to-variable UD source code; moreover, the converse result
shows that there is no UD source code for which the scaled cumulant generating function
of its codeword lengths lies below the R\'{e}nyi entropy.
By invoking L'H\^{o}pital's rule, one gets from \eqref{eq: cumulant generating function}
\begin{align}
\label{eq: limit rho tends to zero}
\lim_{\rho \downarrow 0} \frac{\Lambda_k(\rho)}{\rho}
= \frac1k \sum_{x^k \in \set{X}^k} P_{X^k}(x^k) \, \ell(x^k) = \frac1k \, \expectation[\ell(X^k)].
\end{align}
Hence, by letting $\rho$ tend to zero in \eqref{eq: Campbell's converse result} and
\eqref{eq: Campbell's achievability result}, it follows from \eqref{eq: Shannon entropy}
that Campbell's result in Theorem~\ref{theorem: Campbell} generalizes the well-known bounds
on the optimal average length of UD fixed-to-variable source codes (see, e.g.,
\cite[(5.33) and (5.37)]{Cover_Thomas}):
\begin{align}
\label{eq: Shannon}
\frac{1}{\log D} \; H(X) \leq \frac1k \; \expectation[\ell(X^k)] \leq \frac{1}{\log D} \; H(X) + \frac1k,
\end{align}
and \eqref{eq: Shannon} is satisfied by Huffman coding (see, e.g., \cite[Theorem~5.8.1]{Cover_Thomas}).
Campbell's result therefore generalizes Shannon's fundamental result in \cite{CES48} for the average
codeword lengths of lossless compression codes, expressed in terms of the Shannon entropy.

Following the work by Campbell \cite{Campbell65}, Courtade and Verd\'{u}
derived in \cite{CV2014a} non-asymptotic bounds for the scaled cumulant
generating function of the codeword lengths for $P_X$-optimal variable-length
lossless codes \cite{KYSV14,verdubook}. These bounds were used in \cite{CV2014a}
to obtain simple proofs of the asymptotic normality of the distribution of
codeword lengths, and the reliability function of memoryless sources allowing
countably infinite alphabets. Sason and Verd\'{u} recently derived in \cite{SasonV18b}
improved non-asymptotic bounds on the cumulant generating function of the codeword
lengths for fixed-to-variable optimal lossless source coding without prefix constraints,
and non-asymptotic bounds on the reliability function of a DMS, tightening the bounds
in \cite{CV2014a}.

\subsubsection{Analysis}
\label{analysis-source compression}
The following analysis for lossless source compression with UD codes relies on
a combination of Theorems~\ref{thm:Huffman} and~\ref{theorem: Campbell}.

Let $X_1, \ldots, X_k$ be i.i.d. symbols which are emitted from a DMS according to a probability
mass function $P_X$ whose support is a finite set $\set{X}$ with $|\set{X}|=n$. Similarly to
Section~\ref{subsection: guessing}, in order to cluster the data, suppose that
each symbol $X_i$ is mapped to $Y_i = f(X_i)$ where $f \in \set{F}_{n,m}$ is an arbitrary
deterministic function (independent of the index $i$) with $m<n$. Consequently, the i.i.d.
symbols $Y_1, \ldots, Y_k$ take values on a set $\set{Y}$ with $|\set{Y}|=m<|\set{X}|$.
Consider two UD fixed-to-variable source codes: one operating on the sequences $x^k \in \set{X}^k$,
and the other one operates on the sequences $y^k \in \set{Y}^k$; let $D$ be the size of
the alphabets of both source codes.
Let $\ell(x^k)$ and $\overline{\ell}(y^k)$ denote the length of the codewords for the
source sequences $x^k$ and $y^k$, respectively, and let $\Lambda_k(\cdot)$ and
$\overline{\Lambda}_k(\cdot)$ denote their corresponding scaled cumulant generating functions
(see \eqref{eq: cumulant generating function}).

In view of Theorem~\ref{theorem: Campbell}-b), for every $\rho > 0$,
there exists a UD source code for the sequences in $\set{X}^k$ such that the scaled cumulant
generating function of its codeword lengths satisfies \eqref{eq: Campbell's achievability result}.
Furthermore, from Theorem~\ref{theorem: Campbell}-a), we get
\begin{align}
\label{eq: 20181104-a}
\frac{\overline{\Lambda}_k(\rho)}{\rho} \geq \frac1{\log D} \; H_{\frac1{1+\rho}}(Y).
\end{align}
From \eqref{eq: Campbell's achievability result}, \eqref{eq: 20181104-a}
and Theorem~\ref{thm:Huffman}~a) and b), for every $\rho>0$, there
exist a UD source code for the sequences in $\set{X}^k$, and a construction of a
deterministic function $f \in \set{F}_{n,m}$ (as specified by Steps~1--4 in the
proof of Theorem~\ref{thm:Huffman}, borrowed from \cite{CicaleseGV18}) such that
the difference between the two scaled cumulant generating functions satisfies
\begin{align}
\label{eq:UB lossless}
\Lambda_k(\rho) - \overline{\Lambda}_k(\rho) \leq \frac{\rho}{\log D} \left[ H_{\frac1{1+\rho}}(X)
- H_{\frac1{1+\rho}}(\widetilde{X}_m) + v\hspace*{-0.07cm}\left(\frac1{1+\rho}\right) \right] + \frac{\rho}{k},
\end{align}
where \eqref{eq:UB lossless} holds for every UD source code operating on the sequences in $\set{Y}^k$
with $Y_i = f(X_i)$ (for $i = 1, \ldots, k$) and the specific construction of $f \in \set{F}_{n,m}$ as above, and
$\widetilde{X}_m$ in the right side of \eqref{eq:UB lossless} is a random variable whose probability mass
function is given in \eqref{eq:P_Xm}. The right side of \eqref{eq:UB lossless} can be very well approximated
(for all $\rho>0$) by using \eqref{eq:20181011-i}.

We proceed with a derivation of a lower bound on the left side of \eqref{eq:UB lossless}.
In view of Theorem~\ref{theorem: Campbell}, it follows that \eqref{eq: Campbell's converse result}
is satisfied for every UD source code which operates on the sequences in $\set{X}^k$; furthermore,
Theorems~\ref{thm:Huffman} and~\ref{theorem: Campbell} imply that, for every $f \in \set{F}_{n,m}$,
there exists a UD source code which operates on the sequences in $\set{Y}^k$ such that
\begin{align}
\frac{\overline{\Lambda}_k(\rho)}{\rho} &\leq \frac1{\log D} \; H_{\frac1{1+\rho}}(Y) + \frac1{k}, \label{eq:20181030b} \\
&\leq \frac1{\log D} \; H_{\frac1{1+\rho}}(\widetilde{X}_m) + \frac1{k}, \label{eq:20181030c}
\end{align}
where \eqref{eq:20181030c} is due to \eqref{eq:lemma6} since $Y_i=f(X_i)$ (for $i = 1, \ldots, k$)
with an arbitrary deterministic function $f \in \set{F}_{n,m}$, and $Y_i \sim P_Y$ for every $i$;
hence, from \eqref{eq: Campbell's converse result}, \eqref{eq:20181030b} and \eqref{eq:20181030c},
\begin{align}
\label{eq:LB lossless}
\Lambda_k(\rho) - \overline{\Lambda}_k(\rho) \geq \frac{\rho}{\log D} \left( H_{\frac1{1+\rho}}(X)
- H_{\frac1{1+\rho}}(\widetilde{X}_m) \right) - \frac{\rho}{k}.
\end{align}

We summarize our result as follows.
\begin{theorem}
\label{thm: lossless compression}
Let
\begin{itemize}
\item
$X_1, \ldots, X_k$ be i.i.d. symbols which are emitted from a DMS according to a probability
mass function $P_X$ whose support is a finite set $\set{X}$ with $|\set{X}|=n$;
\item
Each symbol $X_i$ be mapped to $Y_i = f(X_i)$ where $f \in \set{F}_{n,m}$ is the
deterministic function (independent of the index $i$) with $m<n$, as specified by
Steps~1--4 in the proof of Theorem~\ref{thm:Huffman} (borrowed from \cite{CicaleseGV18});
\item
Two UD fixed-to-variable source codes be used: one code encodes the sequences
$x^k \in \set{X}^k$, and the other code encodes their mappings $y^k \in \set{Y}^k$;
let the common size of the alphabets of both codes be $D$;
\item
$\Lambda_k(\cdot)$ and $\overline{\Lambda}_k(\cdot)$ be, respectively, the
scaled cumulant generating functions of the codeword lengths of the $k$-length
sequences in $\set{X}^k$ (see \eqref{eq: cumulant generating function}) and
their mapping to $\set{Y}^k$.
\end{itemize}
Then, for every $\rho > 0$, the following holds for the difference between the scaled
cumulant generating functions $\Lambda_k(\cdot)$ and $\overline{\Lambda}_k(\cdot)$:
\begin{enumerate}[a)]
\item
There exists a UD source code for the sequences in $\set{X}^k$ such that the
upper bound in \eqref{eq:UB lossless} is satisfied for every UD source code
which operates on the sequences in $\set{Y}^k$;
\item
There exists a UD source code for the sequences in $\set{Y}^k$ such that the lower
bound in \eqref{eq:LB lossless} holds for every UD source code for the sequences in
$\set{X}^k$; furthermore, the lower bound in \eqref{eq:LB lossless} holds in general
for every deterministic function $f \in \set{F}_{n,m}$;
\item
The gap between the upper and lower bounds in \eqref{eq:UB lossless} and
\eqref{eq:LB lossless}, respectively, is at most
$\frac{\rho}{\log D} \; v\hspace*{-0.07cm}\left(\frac1{1+\rho}\right) + \frac{2\rho}{k}$
(the function $v \colon (0, \infty) \to (0, \infty)$ is introduced in \eqref{eq:v}),
which is approximately $\frac{0.08607 \rho \, \log_D 2}{1+\rho} + \frac{2\rho}{k}$;
\item
The UD source codes in Items a) and b) for the sequences in $\set{X}^k$ and
$\set{Y}^k$, respectively, can be constructed to be prefix codes by the
algorithm in Remark~\ref{remark: construction of UD codes}.
\end{enumerate}
\end{theorem}

\vspace*{0.1cm}
\begin{remark}[An Algorithm for Theorem~\ref{thm: lossless compression} d)]
\label{remark: construction of UD codes}
A construction of the UD source codes for the sequences in $\set{X}^k$ and $\set{Y}^k$,
whose existence is assured by Theorem~\ref{thm: lossless compression}~a) and~b) respectively,
is obtained by the following algorithm (of three steps) which also constructs them as prefix codes:
\begin{enumerate}[1)]
\item As a preparatory step, we first calculate the probability mass function $P_Y$
from the given probability mass function $P_X$ and the deterministic function
$f \in \set{F}_{n,m}$ which is obtained by Steps~1--4 in the proof of Theorem~\ref{thm:Huffman};
accordingly, $P_Y(y) = \underset{x \in \set{X}: \, f(x)=y}{\sum} P_X(x)$ for all $y \in \set{Y}$.
We then further calculate the probability mass functions for the i.i.d. sequences in $\set{X}^k$
and $\set{Y}^k$ (see \eqref{eq: pmf}); recall that the number of types in $\set{X}^k$
and $\set{Y}^k$ is polynomial in $k$ (being upper bounded by $(k+1)^{n-1}$ and $(k+1)^{m-1}$,
respectively), and the values of these probability mass functions are fixed over each type;
\item The sets of codeword lengths of the two UD source codes, for the sequences in
$\set{X}^k$ and $\set{Y}^k$, can (separately) be designed according to the achievability
proof in Campbell's paper (see \cite[p.~428]{Campbell65}). More explicitly, let $\alpha := \frac1{1+\rho}$;
for all $x^k \in \set{X}^k$, let $\ell(x^k) \in \naturals$ be given by
\begin{align} \label{eq1: Campbell}
\ell(x^k) = \bigl\lceil -\alpha \log_D P_{X^k}(x^k) + \log_D Q_k \big\rceil
\end{align}
with
\begin{align} \label{eq2: Campbell}
Q_k := \sum_{x^k \in \set{X}^k} P_{X^k}^{\alpha}(x^k) = \left( \sum_{x \in \set{X}} P_X^{\alpha}(x) \right)^k,
\end{align}
and let $\overline{\ell}(y^k) \in \naturals$, for all $y^k \in \set{Y}^k$, be given similarly to \eqref{eq1: Campbell}
and \eqref{eq2: Campbell} by replacing $P_X$ with $P_Y$, and $P_{X^k}$ with $P_{Y^k}$. This suggests
codeword lengths for the two codes which fulfil \eqref{eq: Campbell's achievability result} and \eqref{eq:20181030b},
and also both satisfy Kraft's inequality;
\item The separate construction of two prefix codes (a.k.a. instantaneous codes) based on their given sets of
codeword lengths $\{\ell(x^k)\}_{x^k \in \set{X}^k}$ and $\{\overline{\ell}(y^k)\}_{y^k \in \set{Y}^k}$,
as determined in Step~2, is standard (see, e.g., the construction in the proof of \cite[Theorem~5.2.1]{Cover_Thomas}).
\end{enumerate}
\end{remark}

\vspace*{0.1cm}
Theorem~\ref{thm: lossless compression} is of interest since it provides upper and lower
bounds on the reduction in the cumulant generating function of close-to-optimal UD source
codes as a result of clustering data, and Remark~\ref{remark: construction of UD codes} suggests an algorithm to
construct such UD codes which are also prefix codes. For long enough sequences (as $k \to \infty$),
the upper and lower bounds on the difference between the scaled cumulant generating functions
of the suggested source codes for the original and clustered data almost match (see
\eqref{eq:UB lossless} and \eqref{eq:LB lossless}), being roughly equal to
$\rho \left( H_{\frac1{1+\rho}}(X)- H_{\frac1{1+\rho}}(\widetilde{X}_m) \right)$ (with logarithms
on base $D$, which is the alphabet size of the source codes); as $k \to \infty$,
the gap between these upper and lower bounds is less than $0.08607 \log_D 2$.
Furthermore, in view of \eqref{eq: limit rho tends to zero},
\begin{align}
\lim_{\rho \downarrow 0} \frac{\Lambda_k(\rho) - \overline{\Lambda}_k(\rho)}{\rho}
= \frac1k \left( \expectation[\ell(X^k)] - \expectation[\overline{\ell}(Y^k)] \right),
\end{align}
so, it follows from \eqref{eq: Shannon entropy}, \eqref{eq:v}, \eqref{eq:UB lossless}
and \eqref{eq:LB lossless} that the difference between the average code
lengths (normalized by~$k$) of the original and clustered data satisfies
\begin{align}
\label{eq: 20181030e}
\frac{H(X) - H(\widetilde{X}_m)}{\log D} - \frac1k
\leq \frac{\expectation[\ell(X^k)] - \expectation[\overline{\ell}(Y^k)]}{k}
\leq \frac{H(X) - H(\widetilde{X}_m) + 0.08607 \log 2}{\log D},
\end{align}
where the gap between the upper and lower bounds in \eqref{eq: 20181030e} is
equal to $0.08607 \log_D 2 + \frac1k$.

\appendices

\section{Proof of Lemma~\ref{lemma: min RE}}
\label{appendix: proof of min RE lemma}
We first find the extreme values of $p_{\min}$ under the assumption
that $P \in \set{P}_n(\rho)$. If $\frac{p_{\max}}{p_{\min}} = 1$,
then $P$ is the equiprobable distribution on $\set{X}$ and $p_{\min} = \frac1{n}$.
On the other hand, if $\frac{p_{\max}}{p_{\min}} = \rho$, then the
minimal possible value of $p_{\min}$ is obtained when $P$ is the
one-odd-mass distribution with $n-1$ masses equal to $\rho \, p_{\min}$
and a smaller mass equal to $p_{\min}$. The latter case yields
$p_{\min} = \frac1{1+(n-1)\rho}$.

Let $\beta := p_{\min}$, so $\beta$ can get any value in the interval
$\left[ \frac1{1+(n-1)\rho}, \, \frac1n \right] := \Gamma_\rho^{(n)}$.
From Lemma~\ref{lemma: Lemma 6 from CicaleseGV18}, $P \prec Q_{\beta}$
and $Q_{\beta} \in \set{P}_n(\rho)$, and the Schur-concavity of the
R\'{e}nyi entropy yields $H_{\alpha}(P) \geq H_{\alpha}(Q_\beta)$
for all $P \in \set{P}_n(\rho)$ with $p_{\min} = \beta$.
Minimizing $H_{\alpha}(P)$ over $P \in \set{P}_n(\rho)$ can be hence
restricted to minimizing $H_{\alpha}(Q_\beta)$ over $\beta \in \Gamma_\rho^{(n)}$.

\vspace*{-0.1cm}
\section{Proof of Lemma~\ref{lemma: boundedness and monotonocity}}
\label{appendix: proof of boundedness and monotonocity}
The sequence $\{c_{\alpha}^{(n)}(\rho)\}_{n \in \naturals}$ is non-negative since $H_\alpha(P) \leq \log n$
for all $P \in \set{P}_n$. Furthermore, to prove \eqref{eq: bounded},
\begin{align}
0 \leq c_{\alpha}^{(n)}(\rho) &= \log n - \min_{P \in \set{P}_n(\rho)} H_{\alpha}(P) \\
&\leq \log n - \min_{P \in \set{P}_n(\rho)} H_{\infty}(P) \label{eq:lemma1-1} \\
&\leq \log n - \log \frac{n}{\rho} = \log \rho \label{eq:lemma1-2}
\end{align}
where \eqref{eq:lemma1-1} holds since $H_{\alpha}(P)$ is monotonically decreasing in $\alpha$, and \eqref{eq:lemma1-2}
is due to \eqref{eq: RE of infinite order} and $p_{\max} \leq \frac{\rho}{n}$.

Let $\mathrm{U}_n$ denote the equiprobable probability mass function on $\{1, \ldots, n\}$. By the identity
\begin{align} \label{eq: identity RD}
D_{\alpha}(P \| \mathrm{U}_n) = \log n - H_{\alpha}(P),
\end{align}
and since, by Lemma~\ref{lemma: min RE}, $H_{\alpha}(\cdot)$ attains its minimum over the set of probability mass functions $\set{P}_n(\rho)$,
it follows that $D_{\alpha}(\cdot \| \mathrm{U}_n)$ attains its maximum over this set. Let $P^\ast \in \set{P}_n(\rho)$ be the probability
measure which achieves the minimum in $c_{\alpha}^{(n)}(\rho)$ (see \eqref{eq: def. c}), then from \eqref{eq: identity RD}
\begin{align}
c_{\alpha}^{(n)}(\rho) &= \max_{P \in \set{P}_n(\rho)} D_{\alpha}(P \| \mathrm{U}_n) \label{eq: 20181014-a} \\
&= D_{\alpha}(P^\ast \| \mathrm{U}_n).  \label{eq: 20181014-b}
\end{align}
Let $Q^\ast$ be the probability mass function which is defined on $\{1, \ldots, 2n\}$ as follows:
\begin{align} \label{eq: Q ast appendix}
Q^\ast(i) =
\begin{dcases}
\tfrac12 \, P^\ast(i), & i \in \{1, \ldots, n\}, \\
\tfrac12 \, P^\ast(i-n), & i \in \{n+1, \ldots, 2n\}.
\end{dcases}
\end{align}
Since by assumption $P^\ast \in \set{P}_n(\rho)$, it is easy to verify from \eqref{eq: Q ast appendix} that
\begin{align}
Q^\ast \in \set{P}_{2n}(\rho).
\end{align}
Furthermore, from \eqref{eq: Q ast appendix},
\begin{align}
D_{\alpha}(Q^\ast \| U_{2n}) &= \frac1{\alpha-1} \, \log \left( \sum_{i=1}^{2n} \bigl(Q^\ast(i)\bigr)^\alpha \left(\frac1{2n}\right)^{1-\alpha} \right) \label{eq1} \\
&= \frac1{\alpha-1} \log \left( \tfrac12 \, \sum_{i=1}^n \bigl(P^\ast(i)\bigr)^\alpha \left(\frac1{n}\right)^{1-\alpha} + \tfrac12 \sum_{i=n+1}^{2n} \bigl(P^\ast(i-n)\bigr)^\alpha \left(\frac1{n}\right)^{1-\alpha} \right) \label{eq2} \\
&= \frac1{\alpha-1} \log \left( \sum_{i=1}^n \bigl( P^\ast(i) \bigr)^{\alpha} \left(\frac1{n}\right)^{1-\alpha} \right) \label{eq3} \\
&= D_{\alpha}(P^\ast \| U_n). \label{eq4}
\end{align}
Combining \eqref{eq: 20181014-a}--\eqref{eq4} yields
\begin{align}
c_{\alpha}^{(2n)}(\rho) &= \max_{Q \in \set{P}_{2n}(\rho)} D_{\alpha}(Q \| \mathrm{U}_{2n}) \\
&\geq D_{\alpha}(Q^\ast \| \mathrm{U}_{2n}) \\
&= D_{\alpha}(P^\ast \| \mathrm{U}_n)\\
&= c_{\alpha}^{(n)}(\rho),
\end{align}
proving \eqref{eq: monotonicity}. Finally, in view of \eqref{eq: 20181014-a}, $c_{\alpha}^{(n)}(\rho)$ is monotonically increasing in $\alpha$
since so is the R\'{e}nyi divergence of order $\alpha$ (see \cite[Theorem~3]{ErvenH14}).

\vspace*{-0.1cm}
\section{Proof of Lemma~\ref{lemma: c infinity}}
\label{appendix: proof of lemma c_infinity}

From Lemma~\ref{lemma: min RE}, the minimizing distribution of $H_{\alpha}$ is given by $Q_\beta \in \set{P}_n(\rho)$ where
\begin{align}
Q_\beta = \Bigl(\underbrace{ \, \rho \beta, \ldots, \rho \beta}_i, \, 1 - (n+i \rho - i - 1) \beta, \, \underbrace{\beta, \beta, \ldots, \beta}_{n-i-1} \Bigr)
\end{align}
with $\beta \in \left[ \frac1{1+(n-1)\rho}, \frac1n \right]$, and $1 - (n+i \rho - i - 1) \beta \leq \rho \beta \leq \frac{\rho}{n}$. It therefore follows that
the influence of the middle probability mass of $Q_\beta$ on $H_{\alpha}(Q_\beta)$ tends to zero as $n \to \infty$. Therefore, in this asymptotic case,
one can instead minimize $H_{\alpha}(\widetilde{Q}_m)$ where
\begin{align}
\widetilde{Q}_m = \Bigl(\underbrace{ \, \rho \beta, \ldots, \rho \beta}_m, \, \underbrace{\beta, \beta, \ldots, \beta}_{n-m} \Bigr)
\end{align}
with the free parameter $m \in \{1, \ldots, n\}$ and $\beta = \frac{1}{n+m(\rho-1)}$ (so that the total mass of $\widetilde{Q}_m$
is equal to~1).

For $\alpha \in (0,1) \cup (1, \infty)$, straightforward calculation shows that
\begin{align}
H_{\alpha}(\widetilde{Q}_m) &= \frac1{1-\alpha} \,
\log \left(\sum_{j=1}^n \widetilde{Q}_m^{\alpha}(j) \right) \nonumber \\
&= \log n - \frac1{\alpha-1} \log \left( \frac{1 + \frac{m}{n}
\, (\rho^\alpha-1)}{\left( 1 + \frac{m}{n} \, (\rho-1) \right)^{\alpha}} \right),
\end{align}
and by letting $n \to \infty$, the limit of the sequence $\{c_{\alpha}^{(n)}(\rho)\}_{n \in \naturals}$ exists, and it is equal to
\begin{align} \label{eq: c infty - opt.}
c_{\alpha}^{(\infty)}(\rho) &:= \lim_{n \to \infty} c_{\alpha}^{(n)}(\rho) \nonumber \\
&= \lim_{n \to \infty} \left( \log n - \min_{m \in \{1, \ldots, n\}} H_{\alpha}(\widetilde{Q}_m) \right) \nonumber \\
&= \lim_{n \to \infty} \max_{m \in \{1, \ldots, n\}} \left\{ \frac1{\alpha-1} \,
\log \left( \frac{1 + \frac{m}{n} \, (\rho^\alpha-1)}{\left(1+\frac{m}{n} \,
(\rho-1) \right)^{\alpha}} \right) \right\} \nonumber \\
&= \max_{x \in [0,1]} \left\{ \frac1{\alpha-1} \,
\log \left( \frac{1+(\rho^\alpha-1)x}{\bigl(1+(\rho-1)x\bigr)^\alpha} \right) \right\}.
\end{align}
Let $f_{\alpha} \colon [0,1] \to \Reals$ be given by
\begin{align}
f_\alpha(x) = \frac{1+(\rho^\alpha-1)x}{\bigl(1+(\rho-1)x\bigr)^\alpha}, \quad x \in [0,1].
\end{align}
Then, $f_{\alpha}(0) = f_{\alpha}(1) = 1$, and straightforward calculation shows that its
derivative vanishes if and only if
\begin{align} \label{eq: x ast}
x = x^\ast := \frac{1+\alpha (\rho-1) - \rho^\alpha}{(1-\alpha)(\rho-1)(\rho^\alpha-1)}
\end{align}
which, by the mean value theorem,\footnote{We rely here on a specialized version of the
mean value theorem, known as Rolle's theorem, which states that any real-valued differentiable
function that attains equal values at two distinct points must have a point somewhere between
them where the first derivative at this point is zero.} implies (due to the uniqueness of this point)
that $x^\ast \in (0,1)$. Substituting  \eqref{eq: x ast} into \eqref{eq: c infty - opt.}
gives \eqref{eq: c inf. - alpha neq 1}. Taking the limit of \eqref{eq: c inf. - alpha neq 1}
when $\alpha \to \infty$ gives the result in \eqref{eq: lim c inf. - alpha to inf.}.

In the limit where $\alpha \to 1$, the R\'{e}nyi entropy of order $\alpha$ tends to the Shannon
entropy. Hence, letting $\alpha \to 1$ in \eqref{eq: c inf. - alpha neq 1}, it follows that for
the Shannon entropy
\begin{align}
c_1^{(\infty)}(\rho) &= \lim_{\alpha \to 1} c_{\alpha}^{(\infty)}(\rho) \nonumber \\
&= \lim_{\alpha \to 1} \left\{ \frac1{\alpha-1} \, \log \left(1+\frac{1+\alpha \, (\rho-1) - \rho^\alpha}{(1-\alpha)(\rho-1)}\right)
- \frac{\alpha}{\alpha-1} \, \log \left(1+\frac{1+\alpha \, (\rho-1) - \rho^\alpha}{(1-\alpha)(\rho^\alpha-1)}\right) \right\} \nonumber \\
&= \frac{\rho \log \rho}{\rho-1} - \log \mathrm{e} - \log \left( \frac{\rho \log_{\mathrm{e}}\rho}{\rho-1} \right), \label{limit}
\end{align}
where \eqref{limit} follows by invoking L'H\^{o}pital's rule.
This proves \eqref{eq: c inf. for alpha=1}.

From \eqref{eq: bounded}--\eqref{def: c inf.}, we get $0 \leq c_{\alpha}^{(n)}(\rho) \leq c_{\alpha}^{(\infty)}(\rho)$.
Since $c_{\alpha}^{(n)}(\rho)$ is monotonically increasing in $\alpha \in [0, \infty]$, for every $n \in \naturals$, so is $c_{\alpha}^{(\infty)}(\rho)$;
hence, \eqref{eq: lim c inf. - alpha to inf.} yields $c_{\alpha}^{(\infty)}(\rho) \leq \log \rho$. This proves \eqref{eq: tightened bound}.


\bigskip
\begin{thebibliography}{10}

\bibitem{Arikan96}
E. Arikan, ``An inequality on guessing and its application to
sequential decoding,'' {\em IEEE Trans. on Information Theory},
vol.~42, no.~1, pp.~99--105, January 1996.

\bibitem{ArikanM98-1}
E. Arikan and N. Merhav, ``Guessing subject to distortion,'' {\em IEEE Trans. on
Information Theory}, vol.~44, no.~3, pp.~1041--1056, May 1998.

\bibitem{ArikanM98-2}
E. Arikan and N. Merhav, ``Joint source-channel coding and guessing
with application to sequential decoding,'' {\em IEEE Trans. on Information
Theory}, vol.~44, no.~5, pp.~1756--1769, September 1998.

\bibitem{Arimoto73}
S. Arimoto, ``On the converse to the coding theorem for discrete memoryless
channels,'' \emph{IEEE Trans. on Information Theory}, vol.~19, no.~3, pp.~357--359,
May 1973.

\bibitem{Arimoto75}
S. Arimoto, ``Information measures and capacity of order $\alpha$ for
discrete memoryless channels,'' in {\em Topics in Information Theory -
2nd Colloquium}, Keszthely, Hungary, 1975, Colloquia Mathematica Societatis
Jan\'{o}s Bolyai (I. Csisz\'{a}r and P. Elias editors), Amsterdam,
Netherlands: North Holland, vol.~16, pp.~41--52, 1977.

\bibitem{Arnold07}
B. C. Arnold, ``Majorization: Here, there and everywhere,'' {\em Statistical Science},
vol.~22, no.~3, pp.~407--413, August 2007.

\bibitem{ArnoldS18}
B. C. Arnold and J. M. Sarabia, {\em Majorization and the Lorenz Order with Applications
in Applied Mathematics and Economics}, Springer (Statistics for Social and Behavioral
Sciences), 2018.

\bibitem{Ben-Bassat-Raviv}
M. Ben-Bassat and J. Raviv, ``R\'{e}nyi's entropy and probability of error,''
{\em IEEE Trans. on Information Theory}, vol.~24, no.~3, pp.~324--331, May 1978.

\bibitem{Bhatia}
R. Bhatia, {\em Matrix Analysis}, Graduate Texts in Mathematics, Springer, 1997.

\bibitem{Boztas97}
S. Bozta\c{s}, ``Comments on ``An inequality on guessing and its application
to sequential decoding'','' {\em IEEE Trans. on Information Theory}, vol.~43,
no.~6, pp.~2062--2063, November 1997.

\bibitem{BracherHL15}
A. Bracher, E. Hof and A. Lapidoth, ``Guessing attacks on distributed-storage
systems,'' {\em Proceedings of the 2015 IEEE International Symposium on
Information Theory}, pp.~1585--1589, Hong-Kong, China, June 2015.

\bibitem{BracherLP17}
A. Bracher, A. Lapidoth and C. Pfister, ``Distributed task encoding,''
{\em Proceedings of the 2017 IEEE International Symposium on Information Theory},
pp.~1993--1997, Aachen, Germany, June 2017.

\bibitem{BunteL14a}
C. Bunte and A. Lapidoth, ``Encoding tasks and R\'{e}nyi entropy,'' {\em IEEE
Trans. on Information Theory}, vol.~60, no.~9, pp.~5065--5076, September 2014.

\bibitem{BurinS_IT18}
A. Burin and O. Shayevitz, ``Reducing guesswork via an unreliable oracle,''
{\em IEEE Trans. on Information Theory}, vol.~64, no.~11, pp.~6941--6953,
November 2018.

\bibitem{Campbell65}
L. L. Campbell, ``A coding theorem and R\'{e}nyi's entropy,''
{\em Information and Control}, vol.~8, no.~4, pp.~423--429, August 1965.

\bibitem{Campbell66}
L. L. Campbell, ``Definition of entropy by means of a coding problem,''
{\em Probability Theory and Related Fields}, vol. 6, no.~2, pp.~113--118, June 1966.

\bibitem{ChDu13}
M. M. Christiansen and K. R. Duffy, ``Guesswork, large deviations, and Shannon
entropy,'' {\em IEEE Trans. on Information Theory}, vol.~59, no.~2, pp.~796--802,
February 2013.

\bibitem{CicaleseGV04}
F. Cicalese and U. Vaccaro, ``Bounding the average length of optimal source
codes via majorization theory,'' {\em IEEE Trans. on Information Theory},
vol.~50, no.~4, pp.~633-637, April 2004.

\bibitem{CicaleseGV13}
F. Cicalese, L. Gargano, and U. Vaccaro, ``Information theoretic measures of
distances and their econometric applications,'' {\em Proceedings of the 2013
IEEE International Symposium on Information Theory}, pp.~409--413, Istanbul, Turkey,
July 2013.

\bibitem{CicaleseGV17}
F. Cicalese, L. Gargano, and U. Vaccaro, ``How to find a joint probability
distribution with (almost) minimum entropy given the marginals,'' {\em Proceedings
of the 2017 IEEE International Symposium on Information Theory}, pp.~2178--2182,
Aachen, Germany, June 2017.

\bibitem{CicaleseGV18}
F. Cicalese, L. Gargano, and U. Vaccaro, ``Bounds on the entropy of a
function of a random variable and their applications,'' {\em IEEE Trans.
on Information Theory}, vol.~64, no.~4, pp.~2220--2230, April 2018.

\bibitem{CicaleseV_ISIT18}
F. Cicalese and U. Vaccaro, ``Maximum entropy interval aggregations,''
{\em Proceedings of the 2018 IEEE International Symposium on Information Theory},
pp.~1764--1768, Vail, Colorado, USA, June 2018.

\bibitem{CV2014a}
T. Courtade and S. Verd\'{u}, ``Cumulant generating function of codeword
lengths in optimal lossless compression,'' {\em Proceedings of the 2014
IEEE International Symposium on Information Theory}, pp.~2494--2498, Honolulu,
Hawaii, USA, July 2014.

\bibitem{CV2014b}
T. Courtade and S. Verd\'{u}, ``Variable-length lossy compression and channel coding:
Non-asymptotic converses via cumulant generating functions,'' {\em Proceedings of the 2014
IEEE International Symposium on Information Theory}, pp.~2499--2503, Honolulu,
Hawaii, USA, July 2014.

\bibitem{Cover_Thomas}
T. M. Cover and J. A. Thomas, {\em Elements of Information Theory}, Second edition,
John Wiley \& Sons, 2006.

\bibitem{Csiszar95}
I. Csisz\'{a}r, ``Generalized cutoff rates and R\'{e}nyi information measures,''
{\em IEEE Trans. on Information Theory}, vol.~41, no.~1, pp.~26--34, January~1995.

\bibitem{Dalai13}
M. Dalai, ``Lower bounds on the probability of error for classical and classical-quantum
channels,'' {\em IEEE Trans. on Information Theory}, vol.~59, no.~12, pp.~8027--8056,
December 2013.

\bibitem{ErvenH14}
T. van Erven and P. Harremo{\"{e}}s, ``R\'{e}nyi divergence and Kullback-Leibler
divergence,'' {\em IEEE Trans. on Information Theory}, vol.~60, no.~7, pp.~3797--3820,
July 2014.

\bibitem{Gan2007}
G. Gan, C. Ma and J. Wu, {\em Data Clustering: Theory, Algorithms, and Applications},
ASA-SIAM Series on Statistics and Applied Probability, Philadelphia, PA, USA, 2007.

\bibitem{GareyJ79}
M. R. Garey and D. S. Johnson, {\em Computers and Intractability: A Guide to the Theory
of NP-Completness}, W. H. Freedman and Company, New York, USA, 1979.

\bibitem{HanawalS11}
M. K. Hanawal and R. Sundaresan, ``Guessing revisited: a large deviations approach,''
{\em IEEE Trans. on Information Theory}, vol.~57, no.~1, pp.~70--78, January 2011.

\bibitem{HanawalS11b}
M. K. Hanawal and R. Sundaresan, ``The Shannon cipher system with a guessing wiretapper:
general sources,'' {\em IEEE Trans. on Information Theory}, vol.~57, no.~4, pp.~2503--2516,
April 2011.

\bibitem{HLP52}
G. H. Hardy, J. E. Littlewood and G. P\'{o}lya, {\em Inequalities}, second edition,
Cambridge University Press, Cambridge, UK, 1952.

\bibitem{Harremoes04}
P. Harremo{\"{e}}s, ``A new look on majorization,'' {\em Proceedings of the 2004 IEEE
International Symposium on Information Theory and its Applications}, pp.~1422--1425,
Parma, Italy, October 2004.

\bibitem{Harsha10}
P. Harsha, R. Jain, D. McAllester and J. Radhakrishnan, ``The communication complexity
of correlation,'' {\em IEEE Trans. on Information Theory}, vol.~56, no.~1, pp.~438--449,
January 2010.

\bibitem{HT18}
M. Hayashi and V. Y. F. Tan, ``Equivocations, exponents, and second-order coding rates
under various R\'{e}nyi information measures,'' {\em IEEE Trans. on Information Theory},
vol.~63, no.~2, pp.~975--1005, February 2017.

\bibitem{HoY_IT2010}
S. W. Ho and R. W. Yeung, ``The interplay between entropy and variational distance,''
{\em IEEE Trans. on Information Theory}, vol.~56, no.~12, pp.~5906--5929, December~2010.

\bibitem{HoSV_IT2010}
S. W. Ho and S. Verd\'{u}, ``On the interplay between conditional entropy
and error probability,'' {\em IEEE Trans. on Information Theory}, vol.~56,
no.~12, pp.~5930--5942, December~2010.

\bibitem{HoSV-ISIT15}
S. W. Ho and S. Verd\'u, ``Convexity/concavity of the R\'{e}nyi entropy and $\alpha$-mutual
information,'' {\em Proceedings of the 2015 IEEE International Symposium on Information Theory},
pp.~745--749, Hong Kong, China, June 2015.

\bibitem{HornJ}
R. A. Horn and C. R. Johnson, {\em Matrix Analysis}, second edition, Cambridge University Press, 2013.

\bibitem{Huleihel17}
W. Huleihel, S. Salamatian, and M. M\'{e}dard, ``Guessing with limited memory,''
{\em Proceedings of the 2017 IEEE International Symposium on Information Theory},
pp.~2258--2262, Aachen, Germany, June 2017.

\bibitem{Hanly2012}
H. Inaltekin and S. V. Hanly, ``Optimality of binary power control for the single cell uplink,''
{\em IEEE Trans. on Information Theory}, vol.~58, no.~10, pp.~6484--6496, October 2012.

\bibitem{JelSc72}
F. Jelineck and K. S. Schneider, ``On variable-length-to-block coding,'' {\em IEEE Trans. on
Information Theory}, vol.~18, no.~6, pp.~765--774, November 1972.

\bibitem{Joe88}
H. Joe, ``Majorization, entropy and paired comparisons,'' {\em Annals of Statistics}, vol.~16, no.~2,
pp.~915--925, June 1988.

\bibitem{Joe90}
H. Joe, ``Majorization and divergence,'' {\em Journal of Mathematical Analysis and Applications},
vol.~148, no.~2, pp.~287--305, May~1990.

\bibitem{FnT06a}
E. Jorshweick and H. Bosche, ``Majorization and matrix-monotone functions in wireless communications,''
{\em Foundations and Trends in Communications and Information Theory}, vol.~3, no.~6, pp.~553--701, 2006.

\bibitem{Koga13}
H. Koga, ``Characterization of the smooth R\'{e}nyi entropy using majorization,''
{\em Proceedings of the 2013 IEEE Information Theory Workshop}, pp.~604--608,
Seville, Spain, September 2013.

\bibitem{KYSV14}
I. Kontoyiannis and S. Verd\'{u}, ``Optimal lossless data compression: non-asymptotics and
asymptotics,'' {\em IEEE Trans. on Information Theory}, vol.~60, no.~2, pp.~777--795,
February 2014.

\bibitem{Kuzuoka16}
S. Kuzuoka, ``On the smooth R\'enyi entropy and variable-length source coding allowing
errors,'' {\em Proceedings of the 2016 IEEE International Symposium on Information Theory},
pp.~745--749, Barcelona, Spain, July 2016.

\bibitem{Kuzuoka18}
S. Kuzuoka, ``On the conditional smooth R\'{e}nyi entropy and its applications in
guessing and source coding,'' {\em preprint}, October 22, 2018. [Online]. Available at
\url{https://arxiv.org/abs/1810.09070}.

\bibitem{Leditzky16}
F. Leditzky, M. M. Wilde and N. Datta, ``Strong converse theorems using
R\'{e}nyi entropies,'' {\em Journal of Mathematical Physics}, vol.~57,
no.~8, paper no.~082202, pp.~1--33, August 2016.

\bibitem{LiuSV18}
J. Liu and S. Verd\'{u}, ``Rejection sampling and noncausal sampling under moment constraints,''
{\em Proceedings of the 2018 IEEE International Symposium on Information Theory}, pp.~1565--1569,
Vail, Colorado, USA, June 2018.

\bibitem{MarshallOA}
A. W. Marshall, I. Olkin and B. C. Arnold, {\em Inequalities: Theory of Majorization
and Its Applications}, second edition, Springer,~2011.

\bibitem{Massey94}
J. L. Massey, ``Guessing and entropy,'' {\em Proceedings of the 1994 IEEE International Symposium
on Information Theory}, p.~204, Trondheim, Norway, June 1994.

\bibitem{McElieceYu95}
R. J. McEliece and Z. Yu, ``An inequality on entropy,'' {\em Proceedings of the 1995 IEEE
International Symposium on Information Theory}, p.~329, Whistler, Canada, September 1995.

\bibitem{MerhavAr99}
N. Merhav and E. Arikan, ``The Shannon cipher system with a guessing wiretapper,''
{\em IEEE Trans. on Information Theory}, vol.~45, no.~6, pp.~1860--1866, September 1999.

\bibitem{MosonyiO15}
M. Mosonyi and T. Ogawa, ``Quantum hypothesis testing and the operational interpretation
of the quantum R\'{e}nyi relative entropies,'' {\em Communications in Mathematical Physics},
vol.~334, no.~3, pp.~1617--1648, March 2015.

\bibitem{FnT06b}
D. P. Palomar and Y. Jiang, ``MIMO transceiver design via majorization theory,''
{\em Foundations and Trends in Communications and Information Theory}, vol.~3, no.~4--5, pp.~331--551, 2006.

\bibitem{PfisterSu04}
C. E. Pfister and W. G. Sullivan, ``R\'{e}nyi entropy, guesswork moments and large deviations,''
{\em IEEE Trans. on Information Theory}, vol.~50, no.~11, pp.~2794--2800, November 2004.

\bibitem{PolyanskiyV10}
Y. Polyanskiy and S. Verd\'{u}, ``Arimoto channel coding converse and R\'{e}nyi
divergence,'' {\em Proceedings of the Forty-Eighth Annual Allerton Conference on
Communication, Control and Computing}, pp.~1327--1333, Monticello, Illinois,
USA, October~2010.

\bibitem{Puchala13}
Z. Puchala, L. Rudnicki and K. Zyczkowski, ``Majorization entropic uncertainty relations,''
{\em Journal of Physics A: Mathematical and Theoretical}, vol.~46, no.~27, pp.~1--12, June 2013.

\bibitem{Renyientropy}
A. R\'{e}nyi, ``On measures of entropy and information,'' {\em Proceedings
of the 4th Berkeley Symposium on Probability Theory and Mathematical Statistics},
pp.~547--561, Berkeley, California, USA, 1961.

\bibitem{Roventa2015}
I. Roventa, {\em Recent Trends in Majorization Theory and Optimization: Applications to Wireless Communications},
Editura Pro Universitaria \& Universitaria Craiova, 2015.

\bibitem{Salam17}
S. Salamatian, A. Beirami, A. Cohen and M. M\'{e}dard, ``Centralized versus
decentralized multi-agent guesswork,'' {\em Proceedings of the 2017 IEEE
International Symposium on Information Theory}, pp.~2263--2267, Aachen, Germany,
June 2017.

\bibitem{SantisGV01}
A. De Santis, A. G. Gaggia and U. Vaccaro, ``Bounds on entropy in a guessing game,''
{\em IEEE Trans. on Information Theory}, vol.~47, no.~1, pp.~468--473, January 2001.

\bibitem{Sason16}
I. Sason, ``On the R\'{e}nyi divergence, joint range of relative entropies, and a channel
coding theorem,'' {\em IEEE Trans. on Information Theory}, vol.~62, no.~1, pp.~23--34, January~2016.

\bibitem{SasonV18a}
I. Sason and S. Verd\'{u}, ``Arimoto-R\'{e}nyi conditional entropy and Bayesian $M$-ary
hypothesis testing,'' {\em IEEE Trans. on Information Theory}, vol.~64, no.~1, pp.~4--25,
January 2018.

\bibitem{SasonV18b}
I. Sason and S. Verd\'{u}, ``Improved bounds on lossless source coding
and guessing moments via R\'{e}nyi measures,'' {\em IEEE Trans. on
Information Theory}, vol.~64, no.~6, pp.~4323--4346, June 2018.

\bibitem{SezginJ10}
A. Sezgin and E. A. Jorswieck, ``Applications of majorization theory in space-time
cooperative communications,'' in {\em  Cooperative Communications for Improved Wireless
Network Transmission: Framework for Virtual Antenna Array Applications} (Editor: M. Uysal),
Information Science Reference, pp.~429--470, 2010.

\bibitem{CES48}
C. E. Shannon, ``A mathematical theory of communication,'' {\em Bell System Technical
Journal}, vol.~27, pp.~379--423, 623--656, July-October 1948.

\bibitem{Shayevitz_ISIT11}
O. Shayevitz, ``On R\'{e}nyi measures and hypothesis testing,'' {\em Proceedings
of the 2011 IEEE International Symposium on Information Theory}, pp.~800--804,
Saint Petersburg, Russia, August~2011.

\bibitem{Simic09}
S. Simic, ``Jensen's inequality and new entropy bounds,'' {\em Applied Mathematics Letters},
vol.~22, pp.~1262--1265, 2009.

\bibitem{Steele}
J. M. Steele, {\em The Cauchy-Schwarz Master Class}, Cambridge University Press, 2004.

\bibitem{Sundaresan07}
R. Sundaresan, ``Guessing under source uncertainty,'' {\em IEEE Trans. on Information
Theory}, vol.~53, no.~1, pp.~269--287, January 2007.

\bibitem{Sundaresan07b}
R. Sundaresan, ``Guessing based on length functions,'' {\em Proceedings of the 2007
IEEE International Symposium on Information Theory}, pp.~716--719, Nice, France,
June 2007.

\bibitem{TanH18}
V. Y. F. Tan and M. Hayashi, ``Analysis of ramaining uncertainties and exponents under
various conditional R\'{e}nyi entropies,'' {\em IEEE Transactions on Information Theory},
vol.~64, no.~5, pp.~3734--3755, May 2018.

\bibitem{TomamichelH16}
M. Tomamichel and M. Hayashi, ``Operational interpretation of R\'{e}nyi conditional mutual
information via composite hypothesis testing against Markov distributions,'' {\em Proceedings
of the 2016 IEEE International Symposium on Information Theory}, pp.~585--589, Barcelona,
Spain, July 2016.

\bibitem{Tyagi17}
H. Tyagi, ``Coding theorems using R\'{e}nyi information measures,'' {\em Proceedings of
the 2017 IEEE Twenty-Third National Conference on Communications}, pp.~1--6, Chennai,
India, March 2017.

\bibitem{verdubook}
S. Verd\'{u}, \textit{Information Theory}, in preparation.

\bibitem{Pramod1999}
P. Viswanath and V. Anantharam, ``Optimal sequences and sum capacity of synchronous CDMA systems,''
{\em IEEE Trans. on Information Theory}, vol.~45, no.~6, pp.~1984--1993, September~1999.

\bibitem{Pramod1999b}
P. Viswanath, V. Anantharam and D. N. C. Tse, ``Optimal sequences, power control, and user capacity of synchronous
CDMA systems with linear MMSE multiuser receivers,'' {\em IEEE Trans. on Information Theory}, vol.~45, no.~6,
pp.~1968-1983, September~1999.

\bibitem{Pramod2002}
P. Viswanath and V. Anantharam, ``Optimal sequences for CDMA under colored noise: A Schur-saddle function property,''
{\em IEEE Trans. on Information Theory}, vol.~48, no.~6, pp.~1295--1318, June 2002.

\bibitem{Witsenhhausen}
H. S. Witsenhhausen, ``Some aspects of convexity useful in information theory,'' {\em IEEE Trans. on Information Theory},
vol.~26, no.~3, pp.~265--271, May 1980.

\bibitem{XiWZ11}
B. Xi, S. Wang and T. Zhang, ``Schur-convexity on generalized information entropy and its applications,''
{\em Information Computing and Applications}, Lecture Notes in Computer Science (LNCS, vol.~7030), pp.~153--160,
Springer, 2011.

\bibitem{Yona17}
Y. Yona and S. Diggavi, ``The effect of bias on the guesswork of hash functions,''
{\em Proceedings of the 2017 IEEE International Symposium on Information Theory},
pp.~2253--2257, Aachen, Germany, June 2017.

\bibitem{YuT17}
L. Yu and V. Y. F. Tan, ``R\'{e}nyi resolvability and its applications to the wiretap channel,''
{\em Proceedings of the 10th International Conference on Information Theoretic Security (Lecture
Notes in Computer Science, vol.~10681)}, pp.~208--233, Hong Kong, China, November 2017.

\bibitem{YuT18}
L. Yu and V. Y. F. Tan, ``Wyner's common information under R\'{e}nyi divergence measures,''
{\em IEEE Trans. on Information Theory}, vol.~64, no.~5, pp.~3616--3632, May 2018.
\end{thebibliography}
\end{document}